\documentclass[graybox]{svmult}

\usepackage{type1cm}        
\usepackage{makeidx}         
\usepackage{graphicx}        
\usepackage{multicol}        
\usepackage[bottom]{footmisc}

\usepackage{newtxtext}       %
\usepackage{newtxmath}       

\usepackage{arydshln}
\usepackage{mathtools}

\usepackage{hyperref}

\def \x {{\bf x}}
\def \y {{\bf y}}
\def \z {{\bf z}}
\def \a {{\bf a}}
\def \b {{\bf b}}
\def \c {{\bf c}}
\def \u {{\bf u}}
\def \v {{\bf v}}

\makeindex             

\begin{document}

\title*{On Algebraic Approaches for DNA Codes with Multiple Constraints}
\author{Krishna Gopal Benerjee and Manish K Gupta}
\institute{Krishna Gopal Benerjee \at Department of Electrical Engineering, Indian Institute of Technology Kanpur, India \\ \email{kgopal@iitk.ac.in, kg.benerjee@gmail.com}
\and Manish K Gupta \at Kaushalya: the Skill University, Ahemdabad, India \\ \email{mankg@guptalab.org}}
%
%
\maketitle


\abstract{
DNA strings and their properties are widely studied  since last 20 years due to its applications in DNA computing. 
In this area, one designs a set of DNA strings (called DNA code) which satisfies certain thermodynamic and combinatorial constraints such as reverse constraint, reverse-complement constraint, $GC$-content constraint and Hamming constraint. 
However recent applications of DNA codes in DNA data storage resulted in many new constraints on DNA codes such as avoiding tandem repeats constraint (a generalization of non-homopolymer constraint) and avoiding secondary structures constraint. 
Therefore, in this chapter, we introduce DNA codes with recently developed constraints. 
In particular, we discuss reverse, reverse-complement, $GC$-content, Hamming, uncorrelated-correlated, thermodynamic, avoiding tandem repeats and avoiding secondary structures constraints. 
DNA codes are constructed using various approaches such as algebraic, computational, and combinatorial. 
In particular, in algebraic approaches, one uses a finite ring and a map to construct a DNA code. 
Most of such approaches does not yield DNA codes with high Hamming distance. 
In this chapter, we focus on algebraic constructions using maps (usually an isometry on some finite ring) which yields DNA codes with high Hamming distance. 
We focus on non-cyclic DNA codes. 
We briefly discuss various metrics such as Gau distance, Non-Homopolymer distance etc.
We discuss about algebraic constructions of families of DNA codes that satisfy multiple constraints and/or properties. 
Further, we also discuss about algebraic bounds on DNA codes with multiple constraints. 
Finally, we present some open research directions in this area.
}
\tableofcontents

\section{Introduction}
\label{my sec:1}
DeoxyriboNucleic Acid (DNA) is a blue-print of life storing all the instructions for making living species. The basic structure of DNA is given in Fig. \ref{DNA Fig}. It is a robust molecule and has been used in many emerging areas of DNA computing, DNA nanotechnology, DNA origami, Chemical computing and synthetic biology etc. In most of these applications it is required to construct a set of DNA strings (called DNA codes) that are sufficiently dissimilar. This results in a beautiful but tough problem of construction of DNA strings with certain thermodynamic and combinatorial constraints. There are many ways to construct these objects such as computational (algorithmic ways) and mathematical (algebraic and Combinatorial).  This chapter will focus on algebraic ways to construct such DNA codes. The chapter is organised as follows. 

DNA strings and their properties are discussed in Section \ref{my sec:2}. Section \ref{my sec:3} describes various properties and constrains for DNA codes.
Constructions of DNA codes with various properties and constraints are given in Section \ref{my sec:4} using bijective maps.
Then, DNA codes are constructed from binary codes using Non-Homopolymer Map in Section \ref{my sec:5}.
Further, several algebraic bounds are listed in Section \ref{my sec:6}, and finally  some open problems are given in Section \ref{my sec:7}.

\begin{figure}[t]
\centering
\sidecaption[t]
\includegraphics[scale=.6]{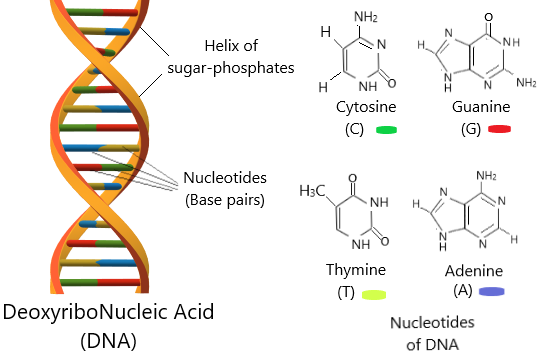}
\caption{DeoxyriboNucleic Acid (DNA) is a double helix structure that is formed by phosphate group, sugar, and four nucleotides (also called bases): Adenine (A), Guanine (G), Cytosine (C), and Thymine (T). 
Adenine and Thymine bind to each other with double hydrogen bond, and similarly, Guanine and Cytosine bind with triple hydrogen bond. 
Thus, Adenine and Thymine, and also, Guanine and Cytosine are Watson-Crick complement to each other. }
\label{DNA Fig}      
\end{figure}

\section{DNA Strings and its Properties}
\label{my sec:2}
In this section, we have defined terms, notations, and properties of DNA strings those are used in this chapter. 
\subsection{DNA Strings}
In this section, we have given formal definitions for string, reverse string, sub-string, concatenated string, DNA string, concatenated DNA string, DNA sub-string in Definition \ref{string def}.
\begin{definition}
For the alphabet $\mathcal{A}_q$ of size $q$ and an integer $n$ ($\geq1$), any one dimensional array $\x$ = $(x_1\ x_2\ \ldots\ x_n)\in\mathcal{A}_q^n$ is called a string of length $n$. 
For any strings $\x$ = $(x_1\ x_2\ \ldots\ x_n)$ over $\mathcal{A}_q$,
\begin{itemize}
\item the reverse string is $\x^r$ = $(x_n\ x_{n-1}\ \ldots\ x_1)$, 
\item for given $1\leq i<j\leq n$, the sub-string is $\x(i,j)$ = $(x_i\ x_{i+1}\ \ldots\ x_j)$, 
\item for given positive integers $i,j,k,l$ ($1\leq i<j\leq n$ and $1\leq k<l\leq n$), two sub-strings $\x(i,j)$ and $\x(k,l)$ are known as disjoint sub-strings of the string $\x$ if $j<k$.
\item for string $\y$ = $(y_1\ y_2\ \ldots\ y_m)$ of length $m$ over $\mathcal{A}_q$, the string 
\[
(\x\ \y) = (x_1\ x_2\ \ldots\  x_n\ y_1\ y_2\ \ldots\ y_m)
\]
of length $n+m$ is called the concatenated string of strings $\x$ and $\y$.
\end{itemize}
\label{string def}
\end{definition}

\begin{example}{Example}
For $q$ = $2$, consider the alphabet $\mathcal{A}_2$ = $\{0,1\}$. 
\begin{itemize}
    \item For the string $\z$ = $(1\ 0\ 0\ 0\ 1\ 1\ 1)$ of length $7$, the reverse string $\z^r$ = $(1\ 1\ 1\ 0\ 0\ 0\ 1)$.
    \item The string $\z(3,6)$ = $(0\ 0\ 1\ 1)$ is a sub-string of the string $\z$ = $(1\ 0\ 0\ 0\ 1\ 1\ 1)$.
    \item The sub-strings $\z(1,3)$ = $(1\ 0\ 0)$ and $\z(5,6)$ = $(1\ 1)$ are disjoint sub-strings of the string $\z$ = $(1\ 0\ 0\ 0\ 1\ 1\ 1)$. 
    \item For $\z_1$ = $(1\ 1\ 1\ 1\ 0)$ and $\z_2$ = $(0\ 0\ 0\ 1)$, the string $(\z_1\ \z_2)$ = $(1\ 1\ 1\ 1\ 0\ 0\ 0\ 0\ 1)$ is the concatenated string of $\z_1$ and $\z_2$.
\end{itemize}
\end{example}
For any string $\x$ of length $n$ over the alphabet $\mathcal{A}_q$, the length of the reverse string $\x^r$ is also $n$.
For any element $a$ in an alphabet $\mathcal{A}_q$ of size $q$, $\a_{r,s}$ is an array of $r$ rows and $s$ columns, $i.e.$,
\[
\a_{r,s} = 
\left(
\begin{array}{cccc}
     a & a & \ldots & a  \\
     a & a & \ldots & a  \\
     \vdots & \vdots & \ddots & \vdots \\
     a & a & \ldots & a  
\end{array}
\right)_{r\times s}.
\]
For the particular case $q$ = $2$, any string and their sub-strings defined over the alphabet $\mathcal{A}_2$ is called binary string and binary sub-strings, respectively. 
Now, we define DNA strings as given in Definition \ref{DNA string def}
\begin{definition}
A DNA string is a string defined over the quaternary alphabet $\Sigma_{DNA}$ = $\{A,C,G,T\}$.
For simplicity, we represent DNA string of length $n$ as $\x$ = $x_1x_2\ldots x_n$. 
For two DNA strings $\x$ and $\y$, the concatenated DNA string of $\x$ and $\y$ is represented by $\textbf{xy}$. 
Similarly, for any DNA string $\x$ = $x_1x_2\ldots x_n$ of length $n$, a sub-string $\x(i,j)$ = $x_ix_{i+1}\ldots x_j$ is called DNA sub-string, where $1\leq i<j\leq n$. 
\label{DNA string def}
\end{definition}
\begin{example}{Example}
Again, the string $AACGAAT\in\Sigma_{DNA}^7$ is a DNA string of length $7$ bps. 
For DNA strings $\x$ = $CACAGT\in\Sigma_{DNA}^6$ and $\y$ = $AAACGCGGG\in\Sigma_{DNA}^9$, strings $\textbf{xy}$ = $CACAGT AAACGCGGG$ and $\textbf{yx}$ = $AAACGCGGG CACAGT$ are concatenated DNA strings each of length $9$ bps.
\end{example}

\subsection{Basic Properties of DNA Strings}
In this section, we have given formal definitions for reverse, reverse-complement and $GC$-weight of any given DNA string.
\begin{definition}
For any given DNA string $\x$ = $x_1x_2\ldots x_n$ of length $n$, 
\begin{itemize}
    \item the \textit{reverse} DNA string is $\x^r$ = $x_nx_{n-1}\ldots x_1$ of length $n$,
    \item the \textit{Watson-Crick complement} or simply \textit{complement} DNA string is $\x^c$ = $x_1^cx_2^c\ldots x_n^c$ of length $n$, and
    \item the \textit{reverse-complement} DNA string is $\x^{rc}$ = $x_n^cx_{n-1}^c\ldots x_1^c$ of length $n$,
\end{itemize}
where $A^c$ = $T$, $C^c$ = $G$, $G^c$ = $C$, and $T^c$ = $A$, $i.e.$, Watson-Crick complement of $A,C,G$ and $T$ are $T,G,C$ and $A$, respectively. 
for simplicity, we call the reverse DNA string and reverse-complement DNA string as R DNA string and RC DNA string, respectively. 
Further, the $GC$-weight of the DNA string $\x$ is the sum of the number of nucleotide $C$ and the number of nucleotide $G$ in the DNA string $\x$.
We denote the $GC$-weight of the DNA string $\x$ by $w_{GC}(\x)$.
\end{definition}

\begin{example}{Example}
For the DNA string $\x$ = $AAGCCAAATC$ of length $10$ bps, 
\begin{itemize}
    \item the reverse DNA string (or R DNA string) is $\x^r$ = $CTAAACCGAA$,
    \item the Watson-Crick complement or complement DNA string is $\x^c$ = $TTCGG$- $TTTAG$, 
    \item the reverse-complement DNA string (or RC DNA string) is $\x^{rc}$ = $GATTTGGCTT$, and 
    \item the $GC$-weight of the DNA string is $w_{GC}(\x)$ = $w_{GC}(AAGCCAAATC)$ = $4$.
\end{itemize}
\end{example}
For several molecular biology techniques, such as designing optimal DNA microarrays, quantitative PCR, and multiplex PCR, DNA hybridization is involved, and it depends on the experimental value of some parameters such as malting temperature \cite{Howley1979,MARMUR1962109,10.1093/bioinformatics/bti066,RUSSELL1969473}.
In \cite{MARMUR1962109}, the melting temperature of a DNA string $\x$ of length $n$ and $GC$-weight $w_{GC}(\x)$ is given by 
\begin{equation}
    T_\x = 64.9+41.0\times\left(\frac{w_{GC}(\x)-16.4}{n}\right).
    \label{DNA temperature 1}
\end{equation}
Further, in \cite{Howley1979}, the salt adjust melting temperature of a DNA string $\x$ of length $n$ and $GC$-weight $w_{GC}(\x)$ is 
\begin{equation}
    T_\x = 100.5+41.0\times\left(\frac{w_{GC}(\x)-36.4}{n}\right)+16.6\log([Na^+]).
    \label{DNA temperature 2}
\end{equation}
Hence, for given length $n$, DNA strings have similar melting temperature if they have similar $GC$-weight.

\subsection{Secondary Structures of DNA strings}
\begin{figure}[t]
\centering
\sidecaption[t]
\includegraphics[scale=.7]{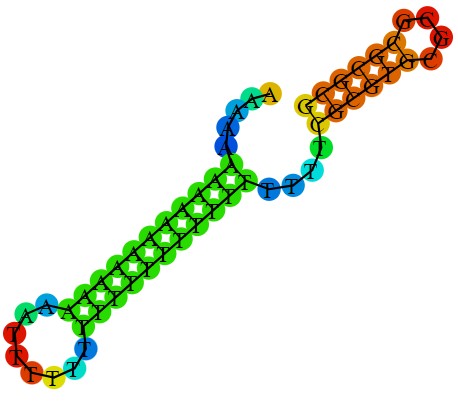}
\caption{Consider the DNA string $\x$ = $AAAAAAAAAAAAAAAAAATTTTTTTTTTTTT$- $TTTTTTTTCGCGTGCGCGCGCGCG$ of length $55$ bps.
The DNA sub-strings $\x(6,16)$ and $\x(40,45)$ bind pairwise with $\x(25,35)^r$ and $\x(50,55)^r$, and it forms two stems one of length $11$ bps and another of length $6$ bps.
The secondary structure for the DNA string $\x$ is predicted by The Vienna RNA Websuite \cite{10.1093/nar/gkn188,Lorenz2011}.}
\label{Sec Str Example Fig}      
\end{figure}
Any chemically active DNA string $\x$ = $x_1x_2\ldots x_n$ of length $n$ form secondary structures by binding upon itself.
An example of such secondary structure in a DNA string is given in Fig. \ref{Sec Str Example Fig}.
Secondary structures can be deduced in a physical DNA using mostly Nuclear Magnetic Resonance (NMR) and X-ray crystallography.
Like all other molecules, DNA must follow the thermodynamic laws, and thus, it is an assumption that the natural fold in any DNA is law energy structure \cite{Deschnes2005AGA}. 
In a given DNA string, secondary structures are \textit{approximately} predicted using a dynamic algorithm known as the Nussinov-Jacobson folding algorithm (NJ algorithm) \cite{Nussinov6309}.
When a DNA string forms a secondary structure then it releases energy called \textit{free energy}, and thus, secondary structures can be predicted by computing the free energy \cite{Clote2000ComputationalMB}.
Further, the free energy can be calculated by computing energies released by binding of $x_i$ with $x_j$ for $i,j\in\{1,2,\ldots,n\}$ and $i<j$. 
The energy released by binding of $x_i$ and $x_j$ are known as \textit{interaction energy} and it is denoted by $\alpha(x_i,x_j)$. 
The assumption for the NJ algorithm is following.
\begin{description}[Assumption]
\item[Assumption] In a DNA string $\x$ = $x_1x_2\ldots x_n$ of length $n$, the interaction energy $\alpha(x_i,x_j)$ between nucleotides $x_i$ and $x_j$ is not depend on all other nucleotide pairs for $1\leq i<j\leq n$.
\end{description}
The interaction energy $\alpha(x_i,x_j)$ is a non-positive value and it depends on the nucleotides $x_i$ and $x_j$.
For any DNA string, the preferable value of interaction energy between $x_i$ and $x_j$ (for details please see \cite{Clote2000ComputationalMB}) is
\begin{equation}
\alpha(x_i,x_j)=\left\{
\begin{array}{rl}
-5 & \mbox{ if } (x_i,x_j)\in\{(G,C),(C,G)\}, \\
-4 & \mbox{ if } (x_i,x_j)\in\{(T,A),(A,T)\}, \\
-1 & \mbox{ if } (x_i,x_j)\in\{(T,G),(G,T)\}, \\
0 & \mbox{ otherwise.}
\end{array}
\right.
\label{interaction energy formula}
\end{equation}
From the assumption, the \textit{minimum free energy}, $E_{i,j}$, for the sub-string $\x(i,j)$ of DNA string $\x$ = $x_1x_2\ldots x_n$ is 
\begin{equation}
  E_{i,j} = \min\left\{E_{i+1,j-1}+\alpha(x_i,x_j), \min_{i<k\leq j}(E_{i,k-1}+E_{k,j})\right\},  
\end{equation}
with the initial conditions $E_{r,r}$ = $0$ and $E_{r-1,r} = 0$ for $r=i,i+1,\ldots,j$ \cite{Clote2000ComputationalMB,10.1007/11779360_9}. 
These initial conditions are followed from the fact that any nucleotide does not interact with itself and immediate neighbours for secondary structures in any DNA.  
For the DNA string $\x$ of length $n$, the minimum free energy is given by $E_{1,n}$. 
A low negative value of $E_{1,n}$ for any DNA string of length $n$ is a good indicator of secondary structures those are exist in the physical DNA.
For any given DNA string, secondary structures are predicted by the RNAfold Web Server using the NJ algorithm \cite{10.1093/nar/gkn188}.
one can observe form the Equation (\ref{interaction energy formula}), any DNA string avoids secondary structures if the DNA string avoids the pairing of $A$ and $T$, the pairing of $G$ and $C$, and the pairing of $G$ and $T$. 
Using the observation, DNA codes that avoids secondary structures are constructed in \cite{9214530}.

From the definition of the interaction energy, we define two terms secondary-complement and reverse-secondary-complement DNA strings as given in Definition \ref{sec-comp reverse-sec-comp definitions}.
\begin{definition}
For any DNA string $\x=x_1x_2\ldots x_n$ of length $n$, 
\begin{itemize}
    \item the \textit{secondary-complement} DNA string $\x^s=x_1^sx_2^s\ldots x_n^s$ of length $n$, and
    \item the \textit{reverse-secondary-complement} DNA string $\x^{rs}=x_n^sx_{n-1}^s\ldots x_1^s$ of length $n$,
\end{itemize}
where $A^s=T$, $T^s=A$, $C^s=G$, $G^s=C$, $G^s=T$ and $T^s=G$.
\label{sec-comp reverse-sec-comp definitions}
\end{definition}

\begin{example}{Example}
Consider the DNA string $\x$ = $ATGAA$ of length $5$ bps.
Then
\begin{itemize}
    \item all the DNA strings $TACTT$, $TGCTT$, $TATTT$ and $TGTTT$ are the secondary-complement DNA strings of $\x$, and
    \item all the DNA strings $TTCAT$, $TTCGT$, $TTTAT$ and $TTTGT$ are the reverse-secondary-complement DNA strings of $\x$.
\end{itemize} 
Note that the secondary-complement and reverse-secondary-complement DNA strings of $\x$ are not unique.
\end{example}
Observe that the secondary-complement of $T$ and $G$ are not unique, and therefore, the secondary complement of any DNA string having the nucleotide $G$ and/or nucleotide $T$ is not unique. 
Also, for any DNA string $\x$,
\begin{itemize}
    \item the DNA string $\x^c$ is a secondary-complement DNA string, and
    \item the DNA string $\x^{rc}$ is a reverse-secondary-complement DNA string.
\end{itemize} 
\begin{proposition}
 If the DNA string $\x$ of length $n$ forms a secondary structure with stem length $\ell$ then there exist two disjoint DNA sub-strings $\x(i,i+\ell-1)$ and $\x(j,j+\ell-1)$ ($i\geq j+\ell$) such that $\x(j,j+\ell-1)=\x(i,i+\ell-1)^s$ or $\x(j,j+\ell-1)=\x(i,i+\ell-1)^{rs}$.
 \label{sec str sec comp proposition}
\end{proposition}
 Now, one can find the following remark.
 \begin{remark}
 Consider a DNA string $\x$ of length $n$ such that the DNA string does not have two sub-strings $\x(i,i+\ell-1)$ and $\x(j,j+\ell-1)$ ($i+\ell< j$) such that $\x(j,j+\ell-1)\neq\x(i,i+\ell-1)^s$ and $\x(j,j+\ell-1)\neq\x(i,i+\ell-1)^{rs}$. Then, the DNA string $\x$ does not form any secondary structure with stems of length $\ell$. 
 \label{Sec Str existence remark}
 \end{remark}
Now, as defined in \cite{Zuker1984}, the secondary structure for any DNA string is defined as following.
\begin{definition}
For any DNA string $\x=x_1x_2\ldots x_n$ of length $n$, consider a set $S$ = $\{\x(i_1,i_2)$, $\x(i_3,i_4)$, $\ldots$, $\x(i_{2j-1},i_{2j})\}$ of DNA sub-strings of $\x$ such that $1\leq i_1<i_2<i_3<\ldots<i_{2j}\leq n$.
A secondary structure is the result of binding pairwise of the nucleotides of DNA sub-strings in the set $S$, $i.e.$, for each $\x(i_s,i_{s+1})\in S$ there exist some $\x(i_t,i_{t+1})\in S$ such that all the nucleotides of the sub-string $\x(i_s,i_{s+1})$ bind pairwise to either the nucleotides of $\x(i_t,i_{t+1})$ or the nucleotides of $\x(i_t,i_{t+1})^r$, where the length of the sub-strings $\x(i_s,i_{s+1})$ and $\x(i_t,i_{t+1})$ are the same, $i.e.$, $i_{t+1}-i+1$ = $i_{s+1}-i_s+1$, and $s,t\in\{1,2,\ldots,2j-1\}$. 
Binding of $\x(i_s,i_{s+1})$ to either $\x(i_t,i_{t+1})$ or $\x(i_t,i_{t+1})^r$ forms stem of length $i_{s+1}-i_s+1$ in the secondary structure for the DNA string $\x$.
Note that every set of DNA sub-strings is not a valid secondary structure, as most possibilities are removed due to chemical and stereochemical constraints. 
\label{Secondary structure def}
\end{definition}
\begin{example}{Example}
As shown in Fig. \ref{Sec Str Example Fig}, for the DNA string $\x$ = $AAAAAAAAAAAAAAAAAA$- $TTTTTTTTTTTTTTTTTTTTTCGCGTGCGCGCGCGCG$ of length $55$ bps, consider $S$ = $\{\x(6,16),$ $\x(25,35),\x(40,45),\x(50,55)\}$, where $\x(6,16)$ = $AAAAAA$- $AAAAA$, $\x(25,35)$ = $TTTTTTTTTTT$, $\x(40,45)$ = $CGCGTG$, and $\x(50,55)$ = $CGCGCG$. 
The DNA sub-strings $\x(6,16)$ of length $11$ bps and $\x(40,45)$ of length $6$ bps bind pairwise with $\x(25,35)^r$ of length $11$ bps and $\x(50,55)^r$ of length $6$ bps, respectively. 
Also, observe that $\x(6,16)$ = $\x(25,35)^{rs}$ and $\x(40,45)$ = $\x(50,55)^{rs}$. 
The secondary structure has two stems of length $11$ bps and $6$ bps.
\end{example}

In Definition \ref{Secondary structure def}, each set of sub-strings of any DNA string is not valid secondary structure, therefore, Proposition \ref{sec str sec comp proposition} is not true in reverse order.
\subsection{Correlations of DNA Strings }
DNA strings can be designed using correlation properties such that the string avoids the forbidden strings or sub-strings.  
In the case of DNA data storage, the block addresses are correspond to forbidden strings in the pool. 
We prefer to design DNA strings in which the part of the information is not encoded into the DNA sub-strings that are the same as the address of any DNA strings. 
This motivates to define self-uncorrelated DNA string and mutually uncorrelated DNA strings \cite{yazdi2015rewritable}.
\begin{definition}
Consider two DNA strings $\x$ and $\y$ of length $n$ and $m$, respectively.
The correlation of $\x$ and $\y$, denoted by $\x\circ\y$, is a binary string $\a$ = $(a_1\ a_2\ \ldots\ a_n)$ of length $n$, where 
\begin{equation*}
    a_i = 
    \begin{cases}
    1 & \mbox{ if }m+i-1<n\mbox{ and }\x(i,m-i-1)=\y(1,m), \\
    0 & \mbox{ if }m+i-1<n\mbox{ and }\x(i,m-i-1)\neq\y(1,m), \\
    1 & \mbox{ if }m+i-1\geq n\mbox{ and }\x(i,n)=\y(1,n-i+1),\mbox{ and } \\
    0 & \mbox{ if }m+i-1\geq n\mbox{ and }\x(i,n)\neq\y(1,n-i+1).
    \end{cases}
\end{equation*}
For any DNA string $\x$ of length $n$, if $\x\circ\x$ = $(1\ \textbf{0}_{1,n-1})$ then the DNA string $\x$ is called self-uncorrelated DNA string.
Further, for any two DNA strings $\x$ of length $n$ and $\y$ of length $m$, if $\x\circ\y$ = $\textbf{0}_{1,n}$ and $\y\circ\x$ = $\textbf{0}_{1,m}$ then the DNA strings $\x$ and $\y$ are called mutually uncorrelated DNA strings.
\end{definition}
\begin{example}{Example}
For DNA strings $\x$ = $ACCATG$ of length $6$ bps and $\y$ = $CATG$ of length $4$ bps, the correlation $\x\circ\y$ = $ACCATG\circ CATG$ = $(0\ 0\ 1\ 0\ 0\ 0)$, where 
\begin{equation*}
    \begin{array}{cccccccccccl}
        \x= & A & C & C & A & T & G &   &   &   &   &  \\
        \y= & C & A & T & G &   &   &   &   &   & \hspace{0.5cm}0 & \hspace{0.5cm}\x(1,4) \neq \y(1,4) \\
            &   & C & A & T & G &   &   &   &   & \hspace{0.5cm}0 & \hspace{0.5cm}\x(2,5) \neq \y(1,4) \\
            &   &   & C & A & T & G &   &   &   & \hspace{0.5cm}1 & \hspace{0.5cm}\x(3,6) = \y(1,4)    \\
            &   &   &   & C & A & T & G &   &   & \hspace{0.5cm}0 & \hspace{0.5cm}\x(4,6) \neq \y(1,3) \\
            &   &   &   &   & C & A & T & G &   & \hspace{0.5cm}0 & \hspace{0.5cm}\x(5,6) \neq \y(1,2) \\
            &   &   &   &   &   & C & A & T & G & \hspace{0.5cm}0 & \hspace{0.5cm}\x(6,6) \neq \y(1,1).
    \end{array}
\end{equation*}
Also, the DNA string $ACAGT$ is self-uncorrelated because $ACAGT\circ ACAGT$ = $(1\ 0\ 0\ 0\ 0)$.
Again, DNA strings $ACAGT$ and $AGCATT$ are mutually uncorrelated because $ACAGT\circ AGCATT$ = $(0\ 0\ 0\ 0\ 0)$ and $AGCATT\circ ACAGT$ = $(0\ 0\ 0\ 0\ 0\ 0)$.
\end{example}
Observe that $\x\circ\y$ and $\y\circ\x$ are not the same in general. 

\section{DNA Codes}
\label{my sec:3}
In this section, we have discussed about the Hamming distance, codes, DNA codes and their properties that helps to reduce cost and errors during synthesis and sequencing physical DNA.

For any positive integers $n$ and $M$, a sub-set $\mathscr{C}\subseteq\mathcal{A}_q^n$ of size $M$ is called a \textit{code} over the alphabet $\mathcal{A}_q$ with the $(n,M,d)$ parameters, where $d$ = $\min\{d(\x,\y):\x,\y\in\mathscr{C}\ s.t.\ \x\neq\y\}$ is called the \textit{minimum distance} and $d(\x,\y)$ is the distance between the strings $\x$ and $\y$ in $\mathcal{A}_q^n$. 
Now, the Hamming distance between $\x$ = $(x_1\ x_2\ \ldots\ x_n)$ and $\y$ = $(y_1\ y_2\ \ldots\ y_n)$ in $\mathcal{A}_q^n$ is the size of the set $\{i:x_i\neq y_i,\mbox{ and }1\leq i\leq n\}$. 
In this chapter, the minimum Hamming distance, and the Hamming distance between $\x$ and $\y$ are denoted by $d_H$ and $H(\x,\y)$, respectively. 
Now, we define the DNA codes in Definition \ref{DNA code definition}.
\begin{definition}
Any $(n,M,d_H)$ code $\mathscr{C}_{DNA}$ defined over the alphabet $\Sigma_{DNA}$ is called DNA code with the minimum Hamming distance $d_H$, the length $n$, and the size $M$. 
\label{DNA code definition}
\end{definition}

\begin{example}{Example}
The set $\mathscr{C}_{DNA}$ = $\{AACC, CCTT, AGGT\}\subset\Sigma_{DNA}^4$ is an $(n=4,M=3,d_H=3)$ DNA code, where $H(AACC, CCTT)$ = $4$,  $H(CCTT, AGGT)$ = $3$, and $H(AACC,$ $AGGT)$ = $3$. 
\end{example}

For any given integer $n$ ($\geq1$), if $\x,\y\in\Sigma_{DNA}^n$ then the following properties are satisfied.
\begin{itemize}
    \item $H(\x,\y^r)$ = $H(\x^r,\y)$.
    \item $H(\x,\y^c)$ = $H(\x^c,\y)$.
    \item $H(\x,\y^{rc})$ = $H(\x^{rc},\y)$ = $H(\x^r,\y^c)$ = $H(\x^c,\y^r)$.
\end{itemize}
\subsection{Constraints on DNA Codes}
\label{subsec:2}

As discussed in \cite{doi:10.1089/10665270152530818}, while reading physical DNA corresponding to DNA string $\x$ in a pool of physical DNA, the non-specific hybridisation can be reduced if, for any physical DNA corresponding to the DNA string $\y$,  
\begin{enumerate}
    \item\label{prop 1} $\x$ and $\y$ are not sufficient similar, 
    \item\label{prop 2} $\x$ and $\y^r$ are not sufficient similar, and 
    \item\label{prop 3} $\x$ and $\y^{rc}$ are not sufficient similar.
\end{enumerate}

The property \ref{prop 1}, motivates to define Hamming constraint with distance parameter $d^*$, that ensures that both the physical DNA strings corresponding to DNA strings $\x$ and $\y$ are differ at at-least $d^*$ positions. 
Formally, Hamming constraint for any DNA code is defined in Section \ref{Hamming Constrain}. 

The property \ref{prop 2}, motivates to define reverse constraint with distance parameter $d^*$.
The reverse constraint ensures the physical DNA string corresponding to DNA string $\x$ is differ with the reverse string of the DNA string corresponding to DNA string $\y$ by at-least $d^*$ positions. 
The reverse constraint for any DNA code is defined in Section \ref{Reverse Constraint}. 

Further, the property $\ref{prop 3}$ indicates that the physical DNA string corresponding to DNA string $\x$ should be differ with the reverse-complement DNA string of the DNA string corresponding to DNA string $\y$ by at-least $d^*$ positions.
It motivates to define the reverse-complement constraint, and the reverse-complement constraint is defined in the Section \ref{RC Constraint}.

\subsubsection{Hamming Constraint with Distance Parameter $d^*$}\label{Hamming Constrain} 
An ($n,M,d_H$) DNA code satisfies the Hamming constraint with the distance parameter $d^*$ if the Hamming distance $H(\x,\y)\geq d^*$ for any DNA codewords $\x,\y\in\mathscr{C}_{DNA}$ and $\x\neq\y$ \cite{doi:10.1089/10665270152530818}.
\begin{example}{Example}
The $(4,3,3)$ DNA code $\mathscr{C}_{DNA}$ = $\{AACC, CCTT, AGGT\}$ satisfies the Hamming constraint with the distance parameter $3$. 
Also, the DNA code satisfies the Hamming constraint with distance parameters $1$ and $2$. 
\end{example}
As given in Definition \ref{DNA code definition}, for any $(n,M,d_H)$ DNA code $\mathscr{C}_{DNA}$
\[
d_H=\min\{H(\x,\y):\x\neq\y\mbox{ and }\x,\y\in\mathscr{C}_{DNA}\},
\]
and thus, $H(\x,\y)\geq d_H$ for each $\x$ and $\y$ in $\mathscr{C}_{DNA}$ such that $\x\neq\y$.
Therefore, the DNA code $\mathscr{C}_{DNA}$ satisfies the Hamming constraint with the distance parameter $d_H$ or simply, we call the property as the Hamming constraint.
Hence, in general, any ($n,M,d_H$) DNA code $\mathscr{C}_{DNA}$ satisfies the Hamming constraint, $i.e.$, $H(\x,\y)\geq d_H$ for $\x,\y\in\mathscr{C}_{DNA}$ and $\x\neq\y$.
Thus, all DNA code discussed in this chapter satisfy the Hamming constraint.

\subsubsection{Reverse Constraint with Distance Parameter $d^*$}\label{Reverse Constraint} 
An ($n,M,d_H$) DNA code $\mathscr{C}_{DNA}$ satisfies reverse constraint with distance parameter $d^*$ if $H(\x^r,\y)\geq d^*$ for any $\x,\y\in\mathscr{C}_{DNA}$ and $\x^r\neq\y$ \cite{doi:10.1089/10665270152530818}.
\begin{description}[\ref{Hamming Constrain}.1]
\item[\ref{Reverse Constraint}.1] Reverse constraint: Any DNA code that satisfies reverse constraint with distance property $d^*$ = $d_H$ is called simply DNA code with reverse constraint, $i.e.$, an ($n,M,d_H$) DNA code $\mathscr{C}_{DNA}$ satisfies reverse constraint if $H(\x^r,\y)\geq d_H$ for any $\x,\y\in\mathscr{C}_{DNA}$ and $\x^r\neq\y$. 
For simplicity, we call the reverse constraint as R constraint in the rest of the chapter.
\end{description}
\begin{example}{Example}
The $(4,3,3)$ DNA code $\mathscr{C}_{DNA}$ = $\{AACC, CCTT, AGGT\}$ satisfies the R constraint with the distance parameter $d^*$ = $2$, where $(AACC)^r$ = $CCAA$, $(CCTT)^r$ = $TTCC$, $(AGGT)^r$ = $TGGA$, and the Hamming distances 
\[
\begin{array}{ll}
    H((AACC)^r,AACC) = 4, & H((CCTT)^r,AACC) = 2, \\
    H((AGGT)^r,AACC) = 4, & H((CCTT)^r,CCTT) = 4, \\
    H((AGGT)^r,CCTT) = 4, & H((AGGT)^r,AGGT) = 2.  
\end{array}
\]
\end{example}

\subsubsection{Reverse-Complement Constraint with Distance Parameter $d^*$}\label{RC Constraint} 
An ($n,M,d_H$) DNA code $\mathscr{C}_{DNA}$ satisfies reverse-complement constraint with distance parameter $d^*$ if $H(\x^{rc},\y)\geq d^*$ for any $\x,\y\in\mathscr{C}_{DNA}$ and $\x^{rc}\neq\y$ \cite{doi:10.1089/10665270152530818}.
\begin{description}[\ref{Hamming Constrain}.1]
\item[\ref{RC Constraint}.1] Reverse-complement constraint: Any DNA code that satisfies reverse-complement constraint with distance property $d^*$ = $d_H$ is called simply DNA code with reverse-complement constraint or RC constraint, $i.e.$, an ($n,M,d_H$) DNA code $\mathscr{C}_{DNA}$ satisfies reverse-complement constraint if $H(\x^{rc},\y)\geq d_H$ for any $\x,\y\in\mathscr{C}_{DNA}$ and $\x^{rc}\neq\y$. 
For simplicity, we call the reverse-complement constraint as RC constraint in the rest of the chapter.
\end{description}
\begin{example}{Example}
The $(4,3,3)$ DNA code $\mathscr{C}_{DNA}$ = $\{AACC, CCTT, AGGT\}$ satisfies the RC constraint with the distance parameter $2$, where $(AACC)^{rc}$ = $GGTT$, $(CCTT)^{rc}$ = $GGAA$, $(AGGT)^{rc}$ = $ACCT$, and the Hamming distances 
\[
\begin{array}{ll}
    H((AACC)^{rc},AACC) = 4, & H((CCTT)^{rc},AACC) = 2, \\
    H((AGGT)^{rc},AACC) = 3, & H((CCTT)^{rc},CCTT) = 4, \\
    H((AGGT)^{rc},CCTT) = 3, & H((AGGT)^{rc},AGGT) = 2.
\end{array}
\]
\end{example}

\subsubsection{Fixed $GC$-Content Constraint with Weight $w$}\label{fixed GC weight}
A DNA code $\mathscr{C}_{DNA}$ satisfies fixed $GC$-content constraint with weight $w$, if $GC$-weight of each DNA string in the DNA code is $w$, $i.e.$, $w_{GC}(\x)$ = $w$ for each $\x\in\mathscr{C}_{DNA}$.
\begin{description}[\ref{fixed GC weight}.1]
    \item[\ref{fixed GC weight}.1] $GC$-content constraint:  
An ($n,M,d_H$) DNA code satisfies $GC$-content constraint if $GC$-content of all DNA strings in the DNA code are same and equal to either $\left\lfloor n/2\right\rfloor$ or $\left\lceil n/2\right\rceil$.
\end{description}
\begin{example}{Example}
The $(4,3,3)$ DNA code $\mathscr{C}_{DNA}$ = $\{AACC, CCTT, AGGT\}$ satisfies the Fixed $GC$-Content constraint with weight $2$. 
Infect, the $(4,3,3)$ DNA code $\mathscr{C}_{DNA}$ also satisfies $GC$-content constraint, since the weight $2$ = $\lfloor 4/2\rfloor$.
\end{example}
From Equation (\ref{DNA temperature 1}) and Equation (\ref{DNA temperature 2}), one can observe that, for given length $n$, the melting temperature of any physical DNA depends on $GC$-weight of the DNA string. 
Therefore, to avoid non-specific hybridisation while sequencing are sequencing physical DNA, DNA strings are preferred those have similar $GC$-weight. 
Also, synthesis and sequencing DNA strings with very high $GC$-weight or very low $GC$-weight pose problems \cite{10.1093/nar/gkn425}.
Again, one can observe that the total number of DNA strings of length $n$ and $GC$-weight $w$ is $\binom{n}{w}2^n$. 
For given $n$, the total number of DNA strings of length $n$ is maximum if $w$ = $\lfloor n/2\rfloor$ or $w$ = $\lceil n/2\rceil$.
So, DNA codes with $GC$-content constraint are preferred.

\subsubsection{Tandem-Free Constraint with Repeat-Length $\ell$}\label{Repeat Free} 
A DNA string $\x$ of length $n$ is called tandem-free DNA string with repeat-length $\ell$ if, for each $m$ = $1,2,\ldots,\ell$, two consecutive sub-strings each of length $m$ are not same, $i.e.$, $\x(i,i+m-1)\neq\x(i+m,i+2m-1)$ for $i=1,2,\ldots,n-2m+1$.
Any DNA code satisfying tandem-free constraint with repeat-length $\ell$, if each DNA codeword of the DNA code is tandem-free DNA strings with repeat-length $\ell$. 
\begin{description}[\ref{Repeat Free}.1]
    \item[\ref{Repeat Free}.1] Homopolymers-free constraint:  
    Any DNA string is called Homopolymers-free, if the DNA string is tandem-free with repeat-length one, $i.e.$, any two nucleotides at consecutive positions are not same. 
    Any DNA code with Homopolymer-free constraint is a DNA code in which all DNA codewords are tandem-free with repeat-length one.
    \item[\ref{Repeat Free}.2] A DNA string is free from Homopolymers of run-length $t$ if there is not exist a DNA sub-string of length $t$ such that all nucleotides of the DNA sub-string are identical. 
\end{description}
\begin{example}{Example}
The DNA string $\x$ = $x_1x_2\ldots x_{12}$ = $TATCTATCAGAT$ is tandem-free with repeat-length $3$, because 
\begin{itemize}
    \item $x_i\neq x_{i+1}$ for $i=1,2,\ldots,11$,
    \item $\x(i,i+1)\neq\x(i+2,i+3)$ for $i=1,2,\ldots,9$, 
    \item $\x(i,i+2)\neq\x(i+3,i+5)$ for $i=1,2,\ldots,7$, but
    \item $\x(1,4)=\x(5,8)$.
\end{itemize}
Further, the $(4,3,3)$ DNA code $\mathscr{C}_{DNA}$ = $\{AACC, CCTT, AGGT\}$ satisfies the tandem-free constraint with repeat-length $3$. 
\end{example}
Some DNA strings can not be synthesised without potential errors such as insertion, deletion and substitution errors. 
For example, DNA strings with Homopolymers of run-length more than two cannot be synthesised without errors. 
Therefore, for large integer $\ell$ ($\geq1$), DNA codes that satisfies tandem-free constraint with repeat-length $\ell$ are preferred.
Again, DNA codes with the Homopolymer-free constraint are also preferred to avoid such potential errors.

\subsubsection{$\ell$-Free Secondary Structures Constraint} 
A DNA string $\x$ of length $n$ is called $\ell$-free secondary structures if there do not exist any two DNA sub-strings $\x(i,i+\ell-1)$ and $\x(j,j+\ell-1)$ such that $\x(i,i+\ell-1)\neq\x(j,j+\ell-1)^s$ and $\x(i,i+\ell-1)\neq\x(j,j+\ell-1)^{rs}$ for each $i\in\{1,2,\ldots,n-2\ell+1\}$, $j\in\{\ell+1,\ell+2,\ldots,n-\ell+1\}$ and $j-i>\ell$. 
An ($n,M,d_H$) DNA code satisfies the $\ell$-free secondary structures constraint if all DNA codewords of the DNA code is free from secondary structures of stem length $\ell$.
\begin{example}{Example}
All the codewords of the $(12,4,4)$ DNA code 
\begin{equation*}
\begin{split}
    \mathscr{C}_{DNA} & = \left\{ACACACACACAC,ACTCTCACTCTC,\right. \\
    &\hspace{4cm}\left. CATCACTCACTC,TCACTCTCACTC\right\}
\end{split}
\end{equation*} 
are $3$-free secondary structures, and therefore, the DNA code satisfy $3$-free secondary structures constraint. 
\end{example}

DNA strings with secondary structures are needed to unfold while reading in wet lab since the DNA is quit slow to react against chemical reagents.
Thus, some additional energy and resources are needed to read the DNA, and it increase the cost. 
Therefore, DNA strings, and thus, DNA codes are preferred that avoids secondary structures. 

\subsubsection{Uncorrelated-Correlated Constraint} 
An ($n,M,d_H$) DNA code $\mathscr{C}_{DNA}$ is called mutually uncorrelated if 
\begin{itemize}
    \item each DNA codeword in $\mathscr{C}_{DNA}$ is self-uncorrelated, $i.e.$, $\x\circ\x$ = $(1\ \textbf{0}_{1,n-1})$ for all $\x\in\mathscr{C}_{DNA}$, and
    \item any two DNA codewords in $\mathscr{C}_{DNA}$ are mutually uncorrelated, $i.e.$, $\x\circ\y$ = $\textbf{0}_{1,n}$ for all $\x,\y\in\mathscr{C}_{DNA}$ and $\x\neq\y$. 
\end{itemize}
\begin{example}{Example}
For the ($5,3,3$) DNA code $\mathscr{C}_{DNA}$ = $\{ACAGT,AGCAT,ACGCG\}$, it can be observed that 
\begin{itemize}
    \item all the correlations $ACAGT\circ ACAGT$, $AGCAT\circ AGCAT$, $ACGCG\circ ACGCG$ are $(1\ \textbf{0}_{1,4})$, and
    \item correlations $ACAGT\circ AGCAT$, $AGCAT\circ ACGCG$, $ACAGT\circ ACGCG$ are $\textbf{0}_{1,5}$.
\end{itemize}  
Thus, the DNA code is mutually uncorrelated. 
\end{example}

\subsubsection{Thermodynamic Constraint} 
A DNA code $\mathscr{C}_{DNA}$ satisfy the thermodynamic constraint if, for given real $\delta\geq0$, 
\begin{equation*}
    |\Delta G_\x-\Delta G_\y|\leq\delta\mbox{ for each }\x,\y\in\mathscr{C}_{DNA},
\end{equation*}
where $|a|$ is the absolute value of the real number $a$, and the terms $\Delta G_\x$ and $\Delta G_\y$ represent the minimum free energy of the DNA strings $\x$ and $\y$, respectively. 
The details are given in \cite{DBLP:journals/corr/LimbachiyaRG16,doi:10.1089/10665270152530818}.

\section{DNA Codes from Bijective Maps and the Hamming Distance}\label{my sec:4}
For any positive integers $q$ and $t$, consider two sets $\mathcal{A}_q$ and $\mathscr{D}\subseteq\Sigma_{DNA}^t$ such that size of both sets are the same and equal to $q$.
Now, consider a bijective map 
\begin{equation}
\varphi:\mathcal{A}_q\rightarrow\mathscr{D}.
\label{Bijective Map}
\end{equation}
For any $\x$ = $(x_1\ x_2\ \ldots\ x_n)\in\mathcal{A}_q^n$, consider $\varphi(\x) = \varphi(x_1)\varphi(x_2)\ldots \varphi(x_n)\in\Sigma_{DNA}^{nt}$. 
For any $\mathscr{C}\subset\mathcal{A}_q^n$, $\varphi(\mathscr{C})=\{\varphi(\x):\mbox{ for each }\x\in\mathscr{C}\}$.
Now, for any $x$ and $y$ in $\mathcal{A}_q$, we define a map 
\begin{equation}
    \begin{split}
        &d:\mathcal{A}_q\times\mathcal{A}_q\rightarrow\mathbb{R} \\ &d(x,y) = H(\varphi(x),\varphi(y)).
    \end{split}
    \label{Bijective hamming distance}
\end{equation}
\begin{lemma}
The map $d:\mathcal{A}_q\times\mathcal{A}_q\rightarrow\mathbb{R}$ such that $d(x,y) = H(\varphi(x),\varphi(y))$, as given in (\ref{Bijective hamming distance}), is a distance.
\end{lemma}
\begin{proof}
From the bijective property of the map $\varphi$ and the distance property of the Hamming distance, one can observe the following. 
\begin{description}[Identity of Indiscernibles:]
    \item[Non Negative Property:] For any $\varphi(x)$ and $\varphi(y)$ in $\mathscr{D}$, $H(\varphi(x),\varphi(y))\geq0$. Therefore, $d(x,y)\geq0$ for any $x,y\in\mathcal{A}_q$.
    \item[Identity of Indiscernibles:] For any $\varphi(x)$ and $\varphi(y)$ in $\mathscr{D}$,
    \begin{equation*}
        \begin{split}
            & H(\varphi(x),\varphi(y))=0 \\
            \Leftrightarrow\hspace{0.1cm} & \varphi(x)=\varphi(y).
        \end{split}
    \end{equation*}
    Thus, from the definitions of map $\varphi$ and the map $d$,  
    \begin{equation*}
        \begin{split}
            & d(x,y)=0 \\
            \Leftrightarrow\hspace{0.1cm} & x=y.
        \end{split}
    \end{equation*}
    \item[Symmetric Property:] For any $\varphi(x)$ and $\varphi(y)$ in $\mathscr{D}$, 
    \[
    H(\varphi(x),\varphi(y)) = H(\varphi(y),\varphi(x)).
    \]
    And therefore, for any $x,y\in\mathcal{A}_q$, \[d(x,y) = d(y,x).\]
    \item[Triangular Property:] For any $\varphi(x)$, $\varphi(y)$ and $\varphi(z)$ in $\mathscr{D}$, 
    \[
    H(\varphi(x),\varphi(z))\leq H(\varphi(x),\varphi(y))+H(\varphi(y),\varphi(z)).
    \]
    This implies, for any $x,y,z\in\mathcal{A}_q$, \[d(x,z)\leq d(x,y)+d(y,z).\] 
\end{description}
Hence, the map given in (\ref{Bijective hamming distance}) is a distance.
\end{proof}
Now, for any $\x$ = $(x_1\ x_2\ \ldots\ x_n)\in\mathcal{A}_q^n$ and $\y$ = $(y_1\ y_2\ \ldots\ y_n)\in\mathcal{A}_q^n$, we define 
\[
d(\x,\y)=\sum_{i=1}^nd(x_i,y_i).
\] 
Now, for any $\x,\y\in\mathcal{A}_q^n$, the distance 
\begin{equation}
    \begin{split}
        d(\x,\y)=&\sum_{i=1}^nd(x_i,y_i) \\
        =&\sum_{i=1}^nH(\varphi(x_i),\varphi(y_i)) \\
        =&H(\varphi(\x),\varphi(\y)).
    \end{split}
    \label{d(x,y)=H(x,y) equation}
\end{equation}

For any code $\mathscr{C}$ over $\mathcal{A}_q$, the minimum distance 
\begin{equation}
d=\min\{d(\x,\y):\x,\y\in\mathscr{C}\mbox{ such that }\x\neq\y\}.
    \label{min d set}
\end{equation}
Now, a relation between the distance for any binary code and the Hamming distance for respective DNA code is given in Lemma \ref{$d$ = $d_H$ binary lemma}.
\begin{lemma}
If the minimum distance is $d$ for any code $\mathscr{C}$ over $\mathcal{A}_q$, and the minimum Hamming distance is $d_H$ for the DNA code $\varphi(\mathscr{C})$ over $\mathscr{D}$, then $d$ = $d_H$.  
\label{$d$ = $d_H$ binary lemma}
\end{lemma}
\begin{proof}
From the bijection property of the map $\varphi:\mathcal{A}_q\rightarrow\mathscr{D}$, the map $\varphi:\mathcal{A}_q^n\rightarrow\mathscr{D}^n$ is also bijective for any integer $n\geq1$. 
Now, from Equation (\ref{d(x,y)=H(x,y) equation}),
\begin{equation*}
    \begin{split}
        & d(\x,\y) = H(\x,\y)\, \mbox{ for any }\x,\y\in\mathcal{A}_q^n\\
        \Rightarrow & \{d(\x,\y):\x\neq\y\mbox{ and }\x,\y\in\mathscr{C}\} \\ 
        & \hspace{2cm} = \{H(\x,\y):\varphi(\x)\neq\varphi(\y)\mbox{ and }\varphi(\x),\varphi(\y)\in\varphi(\mathscr{C})\} \\ 
        \Rightarrow & \min\{d(\x,\y):\x\neq\y\mbox{ and }\x,\y\in\mathscr{C}\} \\ 
        & \hspace{2cm} = \min\{H(\x,\y):\varphi(\x)\neq\varphi(\y)\mbox{ and }\varphi(\x),\varphi(\y)\in\varphi(\mathscr{C})\} \\
        \Rightarrow & d = d_H.
    \end{split}
\end{equation*}
Hence, it follows the proof.
\end{proof}
Thus, one can obtain an isometry as given in Lemma \ref{bijective isometry preposition} as follows. 
\begin{lemma}
    The map $\varphi:(\mathcal{A}_q^n,d) \rightarrow (\mathscr{D}^n, d_H)$ is an isometry. 
\label{bijective isometry preposition}
\end{lemma} 
\begin{proof}
One can find that $d(x,y)$ = $H(\varphi(x),\varphi(y))$ for any $x,y\in\mathcal{A}_q$. 
Thus, for any $\x$ = $(x_1\ x_2\ \ldots\ x_n)$ and $\y$ = $(y_1\ y_2\ \ldots\ y_n)$ in $\mathcal{A}_q^n$, 
\begin{equation*}
    \begin{split}
        d(\x,\y) = & \sum_{i=1}^nd(x_i,y_i) \\
        = & \sum_{i=1}^nH(\varphi(x_i),\varphi(y_i)) \\
        = &\ H(\varphi(\x),\varphi(\y)). 
    \end{split}
\end{equation*}
Thus, the result follows.
\end{proof}
From the distance isometry and the map property, one can get the parameter of DNA code as given in Theorem \ref{parameter map DNA code}.
\begin{theorem}
There exists $(t\cdot n,M,d_H)$ DNA code $\varphi(\mathscr{C})$ for an $(n,M,d)$ code $\mathscr{C}$ over $\mathcal{A}_q$, where $d$ = $d_H$. 
\label{parameter map DNA code}
\end{theorem}
\begin{proof}
Consider an $(n,M,d)$ code $\mathscr{C}$ over $\mathcal{A}_q$. 
The map $\varphi:\mathcal{A}_q\rightarrow\mathscr{D}$ maps an element in $\mathcal{A}_q$ to a DNA string of length $t$, where $\mathscr{D}\subseteq\Sigma_{DNA}^t$. 
Therefore, the DNA codeword length of $\varphi(\mathscr{C})$ is $t\cdot n$. 
From the bijection property of the map $\varphi:\mathcal{A}_q\rightarrow\mathscr{D}$, the size of the DNA code $\varphi(\mathscr{C})$ is the same as the size of the code $\mathscr{C}$, $i.e.$, $M$.
From Lemma \ref{bijective isometry preposition}, the result on distance holds.
\end{proof}
In Lemma \ref{r_rc_exist}, Lemma \ref{closed reverse}, Lemma \ref{closed complement} and Lemma \ref{closed RC}, properties on DNA strings with reverse, complement and reverse-complement DNA strings are given.
\begin{lemma}
For any $\z \in \varphi(\mathscr{C})$, if $\z^r\in\varphi(\mathscr{C})$ and $\z^c\in\varphi(\mathscr{C})$ then $\z^{rc} \in \varphi(\mathscr{C})$ for each $z\in\varphi(\mathscr{C})$.
\label{r_rc_exist}
\end{lemma}
\begin{proof}
For any string $\z$ = $(z_1\ z_2\ \ldots\ z_n)$ in $\varphi(\mathscr{C})$, consider $\z^r$ = $(z_n\ z_{n-1}\ \ldots\ z_1)$ and $\z^c$ = $(z_1^c\ z_2^c\ \ldots\ z_n^c)$ in $\varphi(\mathscr{C})$.
Now, 
\begin{equation*}
    \begin{split}
        & \z = (z_1\ z_2\ \ldots\ z_n) \in \varphi(\mathscr{C}) \\ 
        \implies & \z^r = (z_n\ z_{n-1}\ \ldots\ z_1) \in \varphi(\mathscr{C}) \\
        \implies & (\z^r)^c = (z_n^c\ z_{n-1}^c\ \ldots\ z_1^c) \in \varphi(\mathscr{C}) \\
        \implies & \z^{rc}  \in \varphi(\mathscr{C})
    \end{split}
\end{equation*}
Hence, it follows the result.
\end{proof}

\begin{lemma}
    For any $\x\in\mathscr{C}$, DNA string $\varphi^{-1}(\varphi(\x)^r)\in\mathscr{C}$ if and only if $\z^r\in \varphi(\mathscr{C})$ for each $\z\in\varphi(\mathscr{C})$.
\label{closed reverse}
\end{lemma}
\begin{proof}
For any $\x$ = $(x_1\ x_2\ \ldots\ x_n)\in\mathscr{C}$, consider $\z$ = $\varphi(\x)$ = $\varphi(x_1)\varphi(x_2)\ldots\varphi(x_n)$, and therefore, $\z^r$ = $\varphi(\x)^r$ = $\varphi(x_n)^r\varphi(x_{n-1})^r\ldots\varphi(x_1)^r$. 
Now, 
\begin{equation*}
    \begin{split}
        &\varphi^{-1}(\varphi(\x)^r)\in\mathscr{C} \\
        \Leftrightarrow \hspace{3mm} & \varphi(\x)^r\in\varphi(\mathscr{C}) \\
        \Leftrightarrow \hspace{3mm} & \z^r\in\varphi(\mathscr{C})
    \end{split}
\end{equation*}
Hence, it follows the result.
\end{proof}

\begin{lemma}
For any $\x \in \mathscr{C}$, DNA string $\varphi^{-1}(\varphi(\x)^c)\in\mathscr{C}$ if and only if $\z^c\in\varphi(\mathscr{C})$ for each $z\in\varphi(\mathscr{C})$.
\label{closed complement}
\end{lemma}
\begin{proof}
For any $\x$ = $(x_1\ x_2\ \ldots\ x_n)\in\mathscr{C}$, consider $\z$ = $\varphi(\x)$ = $\varphi(x_1)\varphi(x_2)\ldots\varphi(x_n)$, and therefore, $\z^c$ = $\varphi(\x)^c$ = $\varphi(x_1)^c\varphi(x_2)^c\ldots\varphi(x_n)^c$. 
Now, 
\begin{equation*}
    \begin{split}
        &\varphi^{-1}(\varphi(\x)^c)\in\mathscr{C} \\
        \Leftrightarrow \hspace{3mm} & \varphi(\x)^c\in\varphi(\mathscr{C}) \\
        \Leftrightarrow \hspace{3mm} & \z^c\in\varphi(\mathscr{C})
    \end{split}
\end{equation*}
Hence, it follows the result.
\end{proof}

\begin{lemma}
For any $\x \in \mathscr{C}$, if $\varphi^{-1}(\varphi(\x)^c)\in\mathscr{C}$ and $\varphi^{-1}(\varphi(\x)^r)\in\mathscr{C}$ then $\z^{rc}\in\varphi(\mathscr{C})$ for each $z\in\varphi(\mathscr{C})$.
\label{closed RC}
\end{lemma}
\begin{proof}
For any $\x$ = $(x_1\ x_2\ \ldots\ x_n)\in\mathscr{C}$, consider $\z$ = $\varphi(\x)$ = $\varphi(x_1)\varphi(x_2)\ldots\varphi(x_n)$.
Therefore, 
\[
\z^r = \varphi(\x)^r = \varphi(x_n)^r\varphi(x_{n-1})^r\ldots\varphi(x_1)^r,
\] 
and 
\[
\z^c = \varphi(\x)^c = \varphi(x_1)^c\varphi(x_2)^c\ldots\varphi(x_n)^c.
\] 
Now, 
\begin{equation*}
    \begin{split}
        &\varphi^{-1}(\varphi(\x)^r),\varphi^{-1}(\varphi(\x)^c)\in\mathscr{C} \\
        \Leftrightarrow \hspace{3mm} & \varphi(\x)^r,\varphi(\x)^c\in\varphi(\mathscr{C}) \\
        \Leftrightarrow \hspace{3mm} & (\z^r)^c\in\varphi(\mathscr{C}) \\
        \Leftrightarrow \hspace{3mm} & \z^{rc}\in\varphi(\mathscr{C}) 
    \end{split}
\end{equation*}
Hence, it follows the result.
\end{proof}

\subsection{DNA Codes from the Map for the Ring $\mathbb{Z}_4+u\mathbb{Z}_4$ with $u^2$ = $2+2u$}\label{Sec: Gau theory}
For $t$ = $2$, consider $\mathscr{D}$ = $\Sigma_{DNA}^2$ and $\mathcal{A}_q$ = $\mathbb{Z}_4+u\mathbb{Z}_4$ with $u^2$ = $2+2u$. 
Then, the map as given in (\ref{Bijective Map}) and the distance as shown in (\ref{Bijective hamming distance}) are $Gau$ map and $Gau$ distance, respectively, where $Gau$ map and $Gau$ distance are discussed in \cite{8437313,limbachiyaDNA}.
\subsubsection{The Ring $\mathbb{Z}_4+u\mathbb{Z}_4$ with $u^2$ = $2+2u$}
The ring $\mathbb{Z}_4+u\mathbb{Z}_4$ = $\{a+bu: a,b\in\mathbb{Z}_4\mbox{ and }u^2=2+2u\}$ of size $16$ is the finite commutative local chain ring.
We denote the ring $\mathbb{Z}_4+u\mathbb{Z}_4$ with $u^2$ = $2+2u$ by $R$ in the remaining part of Section \ref{Sec: Gau theory}.
For the ring $R$, zero divisors and unit elements are listed as follows.
\begin{itemize}
    \item Zero divisors: $0$, $2$, $u$, $2+u$, $2u$, $2+2u$, $3u$, $2+3u$, and
    \item Unites: $1$, $3$, $1+u$, $3+u$, $1+2u$, $3+2u$, $1+3u$, $3+3u$.
\end{itemize}
The distinct ideals of the ring are as follows. 
\[
\begin{array}{ll}
\langle 0\rangle &= \{0\}   \\
\langle 2u\rangle &= \{0,2u\} \\ 
\langle 2\rangle &= \langle 2+2u \rangle=\{0,2,2u,2+2u\}   \\ 
\langle u\rangle  &= \langle 2+u\rangle = \langle 3u\rangle = \langle 2+3u\rangle
                            = \{0,2,u,2+u,2u,2+2u,3u,2+3u\} \\ 
\langle1\rangle&=\langle3\rangle=\langle1+u\rangle=\langle3+u\rangle=\langle1+2u\rangle=\langle3+2u\rangle=\langle1+3u\rangle=\langle3+3u\rangle = R \\
\end{array}
\]
Now, for any matrix $G$ with $k$ rows $\textbf{g}_1,\textbf{g}_2\ldots \textbf{g}_k$ over the ring $R$, we denote 
\[
\langle G\rangle =  \left\{\sum_{i=1}^ka_i\textbf{g}_i:a_i\in R\mbox{ for }i=1,2,\ldots,k\right\}.
\]
Any matrix that can be deduced into 
\begin{equation}\label{general_formR}
G=\left(
\begin{array}{ccccc}
I_{k_0} & B_{0,1}   & B_{0,2}   & B_{0,3}   & B_{0,4}   \\
0       & uI_{k_1}  & uB_{1,2}  & uB_{1,3}  & uB_{1,4}  \\
0       & 0         & 2I_{k_2}  & 2B_{2,3}  & 2B_{2,4}  \\
0       & 0         & 0         & 2uI_{k_3} & 2uB_{3,4} \\
\end{array} \right)=\left(
\begin{array}{c}
\textbf{g}_1    \\
\textbf{g}_2         \\
\vdots         \\
\textbf{g}_k        \\
\end{array} \right)
\end{equation} 
is called the matrix of type $\{k_0,k_1,k_2,k_3\}$, 
 where the blocks $B_{i,j}$ ($0\leq i<j\leq 4$) are defined over the ring $R$ and $k$ = $k_0+k_1+k_2+k_3$.
For any matrix of type $\{k_0,k_1,k_2,k_3\}$, the size of $\langle G\rangle$ is $16^{k_0}8^{k_1}4^{k_2}2^{k_3}$ \cite{choie2005codes}.
Any sub-module of $R^n$ is known as a linear code $\mathscr{C}$ over the ring $R$. 
\begin{proposition}
The size of any linear code $\mathscr{C}$ over the ring $R$ with the generator matrix $G$ of type $\{k_0,k_1,k_2,k_3\}$ is $16^{k_0}8^{k_1}4^{k_2}2^{k_3}$.
\label{matrix size}
\end{proposition}
\begin{example}{Example}
Consider the matrix 
\[
G=\left(
\begin{array}{ccccc}
1 & 1 & 1 & 1  & 1  \\
0 & u & u & u  & u  \\
0 & 0 & 2 & 2  & 2  \\
0 & 0 & 0 & 2u & 2u \\
\end{array} \right)
\]
over the ring $R$.
The matrix is of type $\{1, 1, 1, 1\}$, and therefore, the size of the $\langle G\rangle$ is $16^1\cdot8^1\cdot4^1\cdot2^1$ = $1024$, where
\begin{equation*}
    \begin{split}
        \langle G\rangle = & \left\{a_1(1\ 1\ 1\ 1\ 1)+a_2(0\ u\ u\ u\ u)+a_3(0\ 0\ 2\ 2\ 2)+\right. \\ 
        & \hspace{2.5cm}\left. a_4(0\ 0\ 0\ 2u\ 2u): a_1,a_2,a_3,a_4\in R\right\}.
    \end{split}
\end{equation*}
\end{example}

For $\x$ = $(x_1\ x_2\ \ldots\ x_n)\in R^n$, we denote $\x^r$ = $(x_n\ x_{n-1}\ \ldots\ x_1)\in R^n$. 

\subsubsection{The $Gau$ Map}
Consider a bijective map $\varphi_G:R\rightarrow\Sigma_{DNA}^2$ such that Table \ref{Gau Map} holds.
\begin{table}
\caption{The $Gau$ Map.}
\centering
\begin{tabular}{|p{2.8cm}||p{1cm}|p{1cm}|p{1cm}|p{1cm}|p{1cm}|p{1cm}|p{1cm}|p{1cm}|}
\hline
Ring element $x$    & $0$  & $1$    & $2$    & $3$    & $u$  & $1+u$  & $2+u$  & $3+u$ \\ \hline
DNA image $\varphi_G(x)$ & $AA$ & $AG$   & $GG$   & $GA$   & $TG$ & $TA$   & $CA$   & $CG$  \\ \hline\hline
Ring element $x$    & $2u$ & $1+2u$ & $2+2u$ & $3+2u$ & $3u$ & $1+3u$ & $2+3u$ & $3+3u$   \\ \hline
DNA image $\varphi_G(x)$ & $CC$ &  $CT$  & $TT$   & $TC$   & $GT$ & $GC$   & $AC$   & $AT$ \\ \hline 
\end{tabular}
\label{Gau Map}
\end{table}

For any $\x$ = $(x_1\ x_2\ \ldots\ x_n)\in R^n$, consider 
\[
\varphi_G(\x) = \varphi_G(x_1)\varphi_G(x_2)\ldots \varphi_G(x_n)\in\Sigma_{DNA}^{2n}.
\] 

Then, for any $\mathscr{C}\subseteq R^n$, we define 
\[
\varphi(\mathscr{C})=\{\varphi_G(\x):\x\in\mathscr{C}\}.
\]
Now, the properties of the $Gau$ map $\varphi_G$ are as follows.
\begin{enumerate}
    \item Reverse property: For each $x\in R$, $\varphi_G(x)^r$ = $\varphi_G(3x)$.
    \item Complement property: For each $x\in R$, $\varphi_G(x)^c$ = $\varphi_G(x+(2+2u))$.
    \item Reverse-complement property: For each $x\in R$, $\varphi_G(x)^{rc}$ = $\varphi_G(3x+(2+2u))$.
\end{enumerate}
Also, some fundamental $Gau$ map properties are listed in Table \ref{Gau Map property}.
\begin{table}
\caption{Some fundamental properties of the $Gau$ Map.}
\centering
\begin{tabular}{|c|p{5cm}|p{6cm}|}
\hline
Sr. no. & Properties of $x\in R$    & Properties of $\varphi_G(x)\in\Sigma_{DNA}^2$   \\ \hline
 1. & For each $x \in R$, $3x$ is unique & For each $\varphi_G(x) \in \Sigma_{DNA}^2$, $\varphi_G(x)^r$ is unique  \\ \hline
 2. & $x$ = $3x$ for $x=0,2,2u,2+2u$  & $\varphi_G(x)$ = $\varphi_G(x)^r$ \newline for $\varphi_G(x)$ = $AA$, $GG$, $CC$, $TT$  \\ \hline
 3. & For each $x \in R$, $x+(2+2u)$ is unique & For each $\varphi_G(x) \in \Sigma_{DNA}^2$, $\varphi_G(x)^c$ is unique  \\ \hline
 4. & For each $x \in R$, $x\neq x+(2+2u)$ & For each $\varphi_G(x)\in\Sigma_{DNA}^2$, $\varphi_G(x)\neq \varphi_G(x)^c$  \\ \hline
 5. & For each $x \in R$, $3x+(2+2u)$ is unique & For each $\varphi_G(x) \in \Sigma_{DNA}^2$, $\varphi_G(x)^{rc}$ is unique  \\ \hline
 6. & $x$ = $3x+(2+2u)$ \newline for $x=3+3u, 1+u, 3+u, 1+3u$  & $\varphi_G(x)$ = $\varphi_G(x)^{rc}$ \newline for $\varphi_G(x)$ = $AT$, $TA$, $CG$, $GC$  \\ \hline
 7. & There is not exists $x$ in $R$ such that \newline $x$ = $3x+(2+2u)$, and $x$ = $3x$  & There is not exists $\varphi_G(x)$ in $\Sigma_{DNA}^2$ such that \newline $\varphi_G(x)$ = $\varphi_G(x)^{rc}$, and $\varphi_G(x)$ = $\varphi_G(x)^r$  \\ \hline
\end{tabular}
\label{Gau Map property}
\end{table}

\subsubsection{The $Gau$ Distance}
In order to compute the Hamming distance on $\Sigma_{DNA}^2$, as given in (\ref{gaumatrix}), the sixteen elements of the ring are arranged in a square matrix $\mathcal{M}$ = $[m_{i,j}]$ such that 
\begin{equation}
    H(\varphi_G(m_{i,j}),\varphi_G(m_{i',j'})) =
\begin{cases}
0\ \mbox{ if } i=i' \mbox{ and } j=j', \\ 
1\ \mbox{ if } i=i' \mbox{ and } j\neq j', \\ 
1\ \mbox{ if } i\neq i' \mbox{ and } j=j', \mbox{ and } \\ 
2\ \mbox{ if } i\neq i' \mbox{ and } j\neq j'. \\ 
\end{cases}
\label{Gau distance basic property}
\end{equation}
For the ring $R$ and set $\Sigma_{DNA}^2$, the square matrix $\mathcal{M}$ with the property as given in Equation \ref{Gau distance basic property} is not unique, and one of the possible arrangement for the square matrices $\mathcal{M}$ = $[m_{i,j}]$ and $\varphi_G(\mathcal{M})$ = $[\varphi_G(m_{i,j})]$ are 
\begin{equation}
\begin{array}{cc}
    \mathcal{M} = &  \begin{array}{cc}
\begin{array}{cccccccccccccccccccccccccccccccc}
            A & && && & & G && && & & & C & &&& & & & T
        \end{array}  
        & \\
        \left(\begin{array}{cccc}
         0   & 3    & 2+u & 1+u  \\
        1    & 2    & 3+u & u     \\
        2+3u & 1+3u & 2u & 3+2u \\
        3+3u & 3u   & 1+2u & 2+2u
\end{array}\right)
        &
        \begin{array}{c}
           A \\
            G \\
            C \\
            T 
        \end{array}  
    \end{array}
\end{array},
\label{gaumatrix}
\end{equation}
and 
\[
\varphi(\mathcal{M}) = \left(\begin{array}{cccc}
        AA & GA & CA & TA  \\
        AG & GG & CG & TG  \\
        AC & GC & CC & TC  \\
        AT & GT & CT & TT
\end{array}\right).
\]
Thus, Gau distance is defined over the ring $R$ such that these properties are preserved. 

For any $x, y \in R$, there exist $0\leq i,i',j,j'\leq3$ such that let $x = m_{i,j}$ and $y = m_{i',j'}$.
Now, $Gau$ distance is defined as 
\begin{equation} 
d_G(x, y) = \min\{1, i + 3i'(\mbox{mod } 4)\} + \min\{1, j + 3j'(\mbox{mod } 4)\},
\end{equation}
where, one can observe that the terms 
\[
\min\{1, i + 3i'(\mbox{mod }4)\} = 
\begin{cases}
0 & \mbox{ if }i=i', \\
1 & \mbox{ if }i\neq i',
\end{cases}
\]
and 
\[
\min\{1, j + 3j'(\mbox{mod }4)\} = 
\begin{cases}
0 & \mbox{ if }j=j', \\
1 & \mbox{ if }j\neq j'.
\end{cases}
\]
Also, for any two elements $m_{i,j}$ and $m_{i,j'}$ of the matrix $\mathcal{M}$ over the ring $R$, $m_{i,j}$ = $m_{i,j'}$ if and only if $i$ = $i'$ and $j$ = $j'$.
\begin{example}{Example}
For $m_{0,1}$ = $3$ and $m_{3,2}$ = $1+3u$, the Gau distance 
\begin{equation*}
    \begin{split}
d_G(3, 1+3u) = & \min\{1, 0 + 3\cdot 3(\mbox{mod }4)\} + \min\{1, 1 + 3\cdot2(\mbox{mod }4)\} \\
                 = & \min\{1, 1\} + \min\{1, 1 + 3\} \\
                 = & 2
    \end{split}
\end{equation*}
\end{example}
Now, one can establish a distance isometry between the ring $R$ and the set $\Sigma_{DNA}^2$ as given in Theorem \ref{gau preserve theorem}.
\begin{theorem}
(\cite[Theorem 1]{8437313}) The $Gau$ map $\varphi_G:(R^n,d_G) \rightarrow (\Sigma_{DNA}^{2n}, d_H)$ is a distance preserving map. 
\label{gau preserve theorem}
\end{theorem}
\begin{proof}
Using computation, it can be easily observed that, for any $x$ and $y$ in $R$, $d_G(x,y)$ = $H(\varphi(x),\varphi(y))$. 
Therefore, for any $\x$ = $(x_1\ x_2\ \ldots\ x_n)$ and $\y$ = $(y_1\ y_2\ \ldots\ y_n)$ in $R^n$, 
\begin{equation*}
    \begin{split}
        d_G(\x,\y) = & \sum_{i=1}^nd_G(x_i,y_i) \\
        = & \sum_{i=1}^nH(\varphi(x_i),\varphi(y_i)) \\
        = & H(\varphi(\x),\varphi(\y)). 
    \end{split}
\end{equation*}
Thus, the result follows.
\end{proof}
Now, for any $x$ and $y$ in $R$, we define a distance 
\begin{equation}
    \begin{split}
        &d:R\times R\rightarrow\mathbb{R} \\ &d(x,y) = H(\varphi_G(x),\varphi_G(y)).
    \end{split}
    \label{distance for Gau}
\end{equation}
Now, one can observe Lemma \ref{Gau distance = Hamming distance lemma} as follows.
\begin{lemma}
For any $x,y\in R$, $d(x,y)$ = $d_G(x,y)$.
\label{Gau distance = Hamming distance lemma}
\end{lemma}
\begin{proof}
The result follows from Equation (\ref{distance for Gau}) and Theorem \ref{gau preserve theorem}.
\end{proof}

\subsubsection{Properties of $Gau$ Map and $Gau$ Distance}
In this section, we have discussed some conditions on codes defined over the ring $R$ that ensures the reverse and complement properties in the DNA codes obtained using $Gau$ map on the codes. 

A linear property for reverse strings defined over the ring $R$ is given in Lemma \ref{linear on reverse} as follows.
\begin{lemma}
For any $\x,\y\in R^n$, and any $a,b\in R$, 
\[
\varphi_G^{-1}(\varphi_G(a\x+b\y)^r) = a\varphi_G^{-1}(\varphi_G(\x)^r) + b\varphi_G^{-1}(\varphi_G(\y)^r).
\]
\label{linear on reverse}
\end{lemma}
\begin{proof}
For any $\x$ = $(x_1\ x_2\ \ldots\ x_n)$ and $\y$ = $(y_1\ y_2\ \ldots\ y_n)$ in $R^n$ and $a,b\in R$, $\varphi_G(a\x+b\y)$ = $\varphi_G(ax_1+by_1)\varphi_G(ax_2+by_2)\dots \varphi_G(ax_n+by_n)$.
Thus, $\varphi_G(a\x+b\y)^r$ = $\varphi_G(ax_n+by_n)^r\ \varphi_G(ax_{n-1}+by_{n-1})^r\dots \varphi_G(ax_1+by_1)^r$.
Therefore, 
\begin{equation*}
\begin{split}
    \varphi_G^{-1}(\varphi_G(a\x+&b\y)^r) \\ 
    = & (\varphi_G^{-1}(\varphi_G(ax_n+by_n)^r)\ \varphi_G^{-1}(\varphi_G(ax_{n-1}+by_{n-1})^r)\ \dots \\
    &\hspace{5.6cm}\dots\ \varphi_G^{-1}(\varphi_G(ax_1+by_1)^r)) \\
     = & (3ax_n+3by_n\ 3ax_{n-1}+3by_{n-1}\ \dots\ 3ax_1+3by_1) \\
     = & (a(3x_n)+b(3y_n)\ a(3x_{n-1})+b(3y_{n-1})\ \dots\ a(3x_1)+b(3y_1)) \\
     = & (a\varphi_G^{-1}(\varphi_G(x_n)^r)+b\varphi_G^{-1}(\varphi_G(y_n)^r)\ a\varphi_G^{-1}(\varphi_G(x_{n-1})^r) \\ 
        & \hspace{0.7cm}+b\varphi_G^{-1}(\varphi_G(y_{n-1})^r)\ \ldots\  a\varphi_G^{-1}(\varphi_G(x_1)^r)+b\varphi_G^{-1}(\varphi_G(y_1)^r)) \\
     = & a(\varphi_G^{-1}(\varphi_G(x_n)^r)\ \varphi_G^{-1}(\varphi_G(x_{n-1})^r)\ \dots\ \varphi_G^{-1}(\varphi_G(x_1)^r)) \\ 
     &\hspace{1cm}+ b((\varphi_G^{-1}(\varphi_G(y_n)^r)\ \varphi_G^{-1}(\varphi_G(y_{n-1})^r)\dots \varphi_G^{-1}(\varphi_G(y_1)^r))) \\
     = & a\varphi_G^{-1}(\varphi_G(\x)^r) + b\varphi_G^{-1}(\varphi_G(\y)^r).
\end{split}
\end{equation*}
It follows the result.
\end{proof}
\begin{example}{Example}
For $\x$ = $(1\ 1\ 2u)\in R^3$, $\y$ = $(0\ 1\ u)\in R^3$, $a$ = $3u\in R$ and $b$ = $2\in R$,  
\begin{equation*}
    \begin{split}
& a\x+b\y = 3u(1\ 1\ 2u)+2(0\ 1\ u)  \\
& \hspace{1.1cm}= (3u\ 2+3u\ 2u) \hspace{1cm} (\because 6u^2 = 2u^2 = 2(2+2u) = 0) \\
& \varphi_G(a\x+b\y) = GTACCC \\
& \varphi_G(a\x+b\y)^r = CCCATG \\
& \varphi_G^{-1}(\varphi_G(a\x+b\y)^r) = (2u\ 2+u\ u)
    \end{split}
\end{equation*}
On the other hand, 
\begin{equation*}
    \begin{split}
& \varphi_G(\x) = AGAGCC  \\
& \varphi_G(\x)^r = CCGAGA \\
& \varphi_G^{-1}(\varphi_G(\x)^r) = (2u\ 3\ 3) \\
& a\varphi_G^{-1}(\varphi_G(\x)^r) = 3u(2u\ 3\ 3) \\
& \hspace{2.1cm} = (0\ u\ u) 
    \end{split}
\end{equation*}
Similarly,
\[
b\varphi_G^{-1}(\varphi_G(\y)^r) = (2u\ 2\ 0)
\]
Therefore, 
\begin{equation*}
    \begin{split}
        a\varphi_G^{-1}(\varphi_G(\x)^r) + b\varphi_G^{-1}(\varphi_G(\y)^r) & = (0\ u\ u) + (2u\ 2\ 0) \\
        & = (2u\ 2+u\ u).
    \end{split}
\end{equation*}
Hence, it is clear that $\varphi_G^{-1}(\varphi_G(a\x+b\y)^r)$ = $a\varphi_G^{-1}(\varphi_G(\x)^r) + b\varphi_G^{-1}(\varphi_G(\y)^r)$.
\end{example}

Similarly one can generalise the Lemma \ref{linear on reverse} as Proposition \ref{linear on revers gen}. 
\begin{proposition}
For any given positive integer $k$ and $i=1,2,\ldots, k$, if $\x_i\in R^n$, then $\varphi_G^{-1}(\varphi_G(\sum_{i=1}^ka_i\x_i)^r) = \sum_{i=1}^ka_i\varphi_G^{-1}(\varphi_G(\x_i)^r)$, where $a_i \ \in R$.
\label{linear on revers gen}
\end{proposition}
Using the linear property as given in Proposition \ref{linear on revers gen}, a condition on generator matrix for linear code defined over the ring $R$ is obtained that ensures the the R constraint in respective DNA code. 
\begin{lemma}
For any matrix 
\[
G = \left(\begin{array}{c}
      \textbf{g}_1 \\
    \textbf{g}_2 \\
    \vdots \\
    \textbf{g}_k
\end{array}\right)
\]
with $k$ rows $\textbf{g}_1$, $\textbf{g}_2$, $\ldots$, $\textbf{g}_k$ over the ring $R$, the DNA code $\varphi_G(\langle G\rangle)$ contains the R DNA strings of each DNA codewords if and only if $\textbf{g}_i^r \in\langle G\rangle$ for each $i=1,2,\ldots,k$.
\label{closure reverse Gau}
\end{lemma}
\begin{proof}
Consider a matrix $G$ with $k$ rows $\textbf{g}_1,\textbf{g}_2,\ldots,\textbf{g}_k$.
For any $\x\in\langle G\rangle$, there exist some $a_i$ ($i=1,2,\ldots,k$) such that $\x$ = $\sum_{i=1}^k a_i\textbf{g}_i$.
\begin{equation*}
    \begin{split}
    &\varphi_G(\y)\in\varphi_G(\langle G\rangle) \\
    \Leftrightarrow & \y\in\langle G\rangle \\
    \Leftrightarrow & \sum_{i=1}^k a_i\textbf{g}_i\in\langle G\rangle \ \mbox{ for some }a_i\in R\mbox{ and }i=1,2,\ldots,k \\
    \Leftrightarrow & \sum_{i=1}^k a_i\textbf{g}_i^r\in\langle G\rangle\hspace{0.5cm} \mbox{ given } \textbf{g}_i^r \in\langle G\rangle \mbox{ for each }i=1,2,\ldots,k \\
    \Leftrightarrow & \sum_{i=1}^k a_i(3\textbf{g}_i^r)\in\langle G\rangle\hspace{0.5cm}\mbox{ from closer property of }\langle G\rangle \\
    \Leftrightarrow & \sum_{i=1}^k a_i\varphi_G^{-1}(\varphi_G(\textbf{g}_i)^r)\in\langle G\rangle\hspace{0.5cm} \mbox{ from reverse property of Gau map } \\
    \Leftrightarrow & \varphi_G^{-1}\left(\varphi_G\left(\sum_{i=1}^k a_i\x_i\right)^r\right)\in\langle G\rangle\hspace{0.5cm}\mbox{ from Proposition \ref{linear on revers gen}} \\
    \Leftrightarrow & \varphi_G^{-1}(\varphi_G(\y)^r)\in\langle G\rangle \\
    \Leftrightarrow & \varphi_G(\y)^r\in\varphi_G(\langle G\rangle) 
    \end{split}
\end{equation*} 
It follows the result.
\end{proof}
\begin{example}{Example}
For the matrix 
\[
G = \left(\begin{array}{ccc}
      1 & 0 & 3
    \end{array}\right),
\]
$k$ = $1$ and $\textbf{g}_1$ = $(1\ 0\ 3)$. 
Observe that $\textbf{g}_1^r$ = $(3\ 0\ 1)$ = $3(1\ 0\ 3)$ = $3\textbf{g}_1$, and therefore, $\textbf{g}_1^r\in\langle G\rangle$.

Also $\langle G\rangle$ = $\{(0\ 0\ 0)$, $(1\ 0\ 3)$, $(2\ 0\ 2)$, $(3\ 0\ 1)$, $(u\ 0\ 3u)$, $(2u\ 0\ 2u)$, $(3u\ 0\ u)$, $(1+u\ 0\ 3+3u)$, $(2+u\ 0\ 2+3u)$, $(3+u\ 0\ 1+3)$, $(1+2u\ 0\ 3+2u)$, $(2+2u\ 0\ 2+2u)$, $(3+2u\ 0\ 1+2u)$, $(1+3u\ 0\ 3+u)$, $(2+3u\ 0\ 2+u)$, $(3+3u\ 0\ 1+u)\}$. 

Therefore, 
$\varphi(\langle G\rangle)$ = $\{AAAAAA$, $AGAAGA$, $GGAAGG$, $GAAAAG$, $TGAAGT$, $CCAACC$, $GTAATG$, $TAAAAT$, $CAAAAC$, $CGAAGC$, $CTAATC$, $TTAATT$, $TCAACT$, $GCAACG$, $ACAACA$, $ATAATA\}$. 
Note that, for each $\z\in\varphi(\langle G\rangle)$, $\z^r\in\varphi(\langle G\rangle)$.
\end{example}

Now, a condition on the linear code defined over the ring $R$ is discussed in Lemma \ref{closure reverse comp Gau} as follows. 
\begin{lemma}
For any given matrix $G$ over the ring $R$, consider the DNA code $\varphi_G(\langle G\rangle)$.
Then, for each $\x\in\varphi_G(\langle G\rangle)$, $\x^c\in\varphi_G(\langle G\rangle)$ if and only if $\textbf{2+2u}_{1,n}\in\langle G\rangle$.  
\label{closure reverse comp Gau}
\end{lemma}
\begin{proof}
For any $\x$ = $(x_1\ x_2\ \ldots\ x_n)\in\langle G\rangle$, if $\textbf{2+2u}_{1,n} = (2+2u\ 2+2u\ \ldots \ 2+2u) \in \varphi_G(\langle G\rangle)$, then
\begin{equation*}
    \begin{split}
        & \varphi_G(\x) \in \varphi_G(\langle G\rangle) \\
        \Rightarrow & \x \in \langle G\rangle \\
        \Rightarrow & \x+(\textbf{2+2u}_{1,n}) \in \langle G\rangle \\
        \Rightarrow & (x_1\ x_2\ \ldots\ x_n)+(2+2u\ 2+2u\ \ldots\ 2+2u) \in \langle G\rangle \\
        \Rightarrow & (x_1+(2+2u)\ x_2+(2+2u)\ \ldots\ x_n+(2+2u)) \in \langle G\rangle \\
        \Rightarrow & \varphi_G(x_1+(2+2u)\ x_2+(2+2u)\ \ldots\ x_n+(2+2u)) \in \varphi_G(\langle G\rangle) \\ 
        \Rightarrow & \varphi_G(x_1+(2+2u))\varphi_G(x_2+(2+2u))\ldots\varphi_G(x_n+(2+2u)) \in \varphi_G(\langle G\rangle) \\ 
        \Rightarrow & \varphi_G(x_1)^c\varphi_G(x_2)^c\ldots\varphi_G(x_n)^c\in\varphi_G(\langle G\rangle)\\
        \Rightarrow & \varphi_G(\x)^c\in\varphi_G(\langle G\rangle).
    \end{split}
\end{equation*} 
For the other side, if $\varphi_G(\x)^c\in\varphi_G(\langle G\rangle)$ for any $\varphi_G(\x)\in\varphi_G(\langle G\rangle)$ then 
\begin{equation*}
    \begin{split}
        & \varphi_G(\textbf{0}_{1,n})\in\varphi_G(\langle G\rangle)\hspace{0.5cm}\mbox{ for particular } \textbf{0}_{1,n} = (0\ 0\ \ldots\ 0)\in\langle G\rangle\\
        \Rightarrow & \varphi_G(\textbf{0}_{1,n})^c\in\varphi_G(\langle G\rangle) \\
        & \hspace{3cm} \varphi_G(\x)^c\in\varphi_G(\langle G\rangle) \mbox{ for any } \varphi_G(\x)\in\varphi_G(\langle G\rangle)\\
        \Rightarrow & \varphi_G(0)^c\varphi_G(0)^c\ldots\varphi_G(0)^c\in\varphi_G(\langle G\rangle)\\
        \Rightarrow & \varphi_G(0+(2+2u))\varphi_G(0+(2+2u))\ldots\varphi_G(0+(2+2u))\in\varphi_G(\langle G\rangle)\\
        \Rightarrow & \varphi_G(2+2u)\varphi_G(2+2u)\ldots\varphi_G(2+2u)\in\varphi_G(\langle G\rangle)\\
        \Rightarrow & \varphi_G(\textbf{2+2u}_{1,n}) \in \varphi_G(\langle G\rangle) \\
        \Rightarrow & \textbf{2+2u}_{1,n}\in\langle G\rangle 
    \end{split}
\end{equation*} 
It follows the result.
\end{proof}
\begin{example}{Example}
For the matrix 
\[
G = \left(\begin{array}{ccc}
      u & u & u
    \end{array}\right),
\]
$k$ = $1$ and $\textbf{g}_1$ = $(u\ u\ u)$. 
Observe that $u\textbf{g}_1$ = $u(u\ u\ u)$ = $(2+2u\ 2+2u\ 2+2u)$, where $u^2$ = $2+2u$. 
But, $u\textbf{g}_1\in\langle G\rangle$, and thus, $(2+2u\ 2+2u\ 2+2u)\in\langle G\rangle$

Also note $\langle G\rangle$ = $\{(0\ 0\ 0)$, $(2\ 2\ 2)$, $(u\ u\ u)$, $(2+u\ 2+u\ 2+u)$, $(2u\ 2u\ 2u)$, $(2+2u\ 2+2u\ 2+2u)$, $(3u\ 3u\ 3u)$, $(2+3u\ 2+3u\ 2+3u)\}$. 

Therefore, 
$\varphi(\langle G\rangle)$ = $\{AAAAAA$, $GGGGGG$, $TGTGTG$, $CACACA$, $CCCCCC$, $TTTTTT$, $GTGTGT$, $ACACAC\}$. 
Note that, for each $\z\in\varphi(\langle G\rangle)$, $\z^{rc}\in\varphi(\langle G\rangle)$.
\end{example}

Now, the parameter of DNA codes obtained from the codes over the ring $R$ using the $Gau$ map are calculated in Theorem \ref{distance preserving Gau} as follows. 
\begin{theorem}
There is an $(2n,M,d_H)$ DNA code $\varphi_G(\mathscr{C})$ for any $(n,M,d_G)$ code $\mathscr{C}$ over the ring $R$, where $d_H$ = $d_G$.
\label{distance preserving Gau}
\end{theorem}
\begin{proof}
The result on length of the DNA codeword follows from the fact that, for any $\x\in R^n$, $\varphi_G(\x)\in\Sigma_{DNA}^{2n}$.
Similarly, the result on the size of the DNA code follows from the fact the $Gau$ map $\varphi_G$ is bijective. 
And, the result on distance follows from Lemma \ref{Gau distance = Hamming distance lemma}.  
\end{proof}

\subsubsection{Constructions of DNA Codes}
Motivated from the $r^{th}$ order binary Reed Muller code, DNA codes are constructed from Reed Muller type code over the ring $R$.
For any integers $r$, $m$ ($0 \leq r \leq m$) and any given element $z\in R$, the generator matrix of the code $\mathcal{R}(r,m,z)$ over the ring $R$ is 
\[
G_{r,m,z} =
\begin{pmatrix}
G_{r,m-1,z} & G_{r,m-1,z}    \\ 
0 & G_{r-1,m-1,z} \\
\end{pmatrix},\ 1 \leq r \leq m-1.
\]
with
\[
 G_{m,m,z}\ = 
 \left(
 \begin{array}{c}
      G_{m-1,m,z}  \\
      0\ 0\ldots\ 0\ z
 \end{array}
 \right)
 \]
and $G_{0,m,z}$ = $\textbf{1}_{1,2^m}$.
Now, in Lemma \ref{Parameter RM code Gau}, the parameter of the $r^{th}$ order Reed Muller type code $\mathcal{R}(r,m,z)$ is calculated. 
\begin{lemma}
Consider the $r^{th}$ order Reed Muller type code $\mathcal{R}(r,m,z)$ with the $(n,M,d_G)$ parameter over the ring $R$.
Then,  
\begin{itemize}
    \item the length 
    \[
    n = 2^m
    \]
    \item the size \[
M = \left \{ \,
\begin{array}{ll}
2^{\left(4\sum_{i=0}^r\binom{m}{i}-3\sum_{i=0}^{r-1}\binom{m-1}{i}\right)} & \mbox{ if } z \in \{2u\},\\
2^{\left(4\sum_{i=0}^r\binom{m}{i}-2\sum_{i=0}^{r-1}\binom{m-1}{i}\right)} & \mbox{ if } z \in \{2,2+2u\}, \\
2^{\left(4\sum_{i=0}^r\binom{m}{i}-\sum_{i=0}^{r-1}\binom{m-1}{i}\right)} & \mbox{ if }  z\in \{u,2+u,3u,2+3u\}, \\
2^{\left(4\sum_{i=0}^r\binom{m}{i}\right)} & \mbox{ if }  z\mbox{ is a unit element of the ring }R, 
\end{array}
\right.\\
\]
and
    \item the minimum Gau distance 
    \[
d_G = \left \{ \,
\begin{array}{ll}
2^{m-r+1} & \mbox{ if } z \in \{2u,2,2+2u\},\\
2^{m-r} & \mbox{ if } z \in R\backslash\{0,2u,2,2+2u\}. 
\end{array}
\right.
\]
\end{itemize}
\label{Parameter RM code Gau}
\end{lemma}
\begin{proof}
For the generator matrix $G_{r,m,z}$, if we denote the number of columns in the matrix $G_{r,m,z}$ by $\ell(G_{r,m,z})$ then, from the generator matrix $G_{r,m,z}$, $\ell(G_{r,m,z})$ = $2\ell(G_{r,m-1,z})$ with the condition $\ell(G_{0,m,z})$ = $2^m$ and $\ell(G_{m,m,z})$ = $\ell(G_{m-1,m,z})$.
after solving the difference equation, we have $\ell(G_{r,m,z})$ = $2^m$, and it follows the result on size of the code $\mathcal{R}(r,m,z)$. 
Note, the total number of rows of the matrix $G_{r,m,z}$ is $\sum_{i=0}^r\binom{m}{i}$.
Also, all the nonzero entry of any given row of the generator matrix $G_{r,m,z}$ are same and it is either $1$ or the element $z$. 
From recurrence, one can calculate that the total number of rows containing the element $z$ is $\sum_{i=0}^{r-1}\binom{m-1}{i}$.
Thus, the matrix $G_{r,m,z}$ is of
\begin{itemize}
    \item type $\left\{\sum_{i=0}^r\binom{m}{i}-\sum_{i=0}^{r-1}\binom{m-1}{i},0,0,\sum_{i=0}^{r-1}\binom{m-1}{i}\right\}$ for $z \in \{2u\}$,
    \item type $\left\{\sum_{i=0}^r\binom{m}{i}-\sum_{i=0}^{r-1}\binom{m-1}{i},0,\sum_{i=0}^{r-1}\binom{m-1}{i},0\right\}$ for $z \in \{2,2+2u\}$, 
    \item type $\left\{\sum_{i=0}^r\binom{m}{i}-\sum_{i=0}^{r-1}\binom{m-1}{i},\sum_{i=0}^{r-1}\binom{m-1}{i},0,0\right\}$ for $z \in \{u,2+u,3u,2+3u\}$, and
    \item type $\left\{\sum_{i=0}^r\binom{m}{i},0,0,0\right\}$ for any unit element $z$ in the ring $R$. 
\end{itemize}  
Hence, the result on code size holds from Proposition \ref{matrix size}.
Now, from symmetry of the matrix $G_{r,m,z}$, any two codewords in $\mathcal{R}(r,m,z)$ are differ at least at $2^{m-r}$ positions. 
Therefore, if $d_z$ = $\min\{d_G(x,y):x\in R\mbox{ and }y\in\langle z\rangle\}$ then the minimum $Gau$ distance $d_G\geq2^{m-r}d_z$, since $d_G(x,y)\geq H(x,y)$ for any $x,y\in R$.
Consider two codewords $\textbf{0}_{1,2^m}$, all zero codeword, and $\textbf{0}_{1,2^m-r}\ \z_{1,r})$, last $r$ positions are $z$ and remaining are zero, in $\mathcal{R}(r,m,z)$.
Then, the $Gau$ distance between these two codewords are $2^{m-r}d_z$, since $d_z\geq1$. 
Thus, from the bound $d_G\geq2^{m-r}d_z$, $d_G$ = $2^{m-r}d_z$. 
Hence, it follows the result on distance for various $z$.
\end{proof}

Now, the properties of the $r^{th}$ order Reed Muller type code $\mathcal{R}(r,m,z)$ is given in Lemma \ref{R and RC RM code Gau}. 
\begin{lemma}
The $r^{th}$ order Reed Muller type code $\mathcal{R}(r,m,z)$ with the generator matrix $G_{r,m,z}$ satisfies
\begin{itemize}
    \item $\textbf{2+2u}_{1,2^m}\in\langle G_{r,m,z}\rangle$, and
    \item $\textbf{g}_i^r\in\langle G_{r,m,z}\rangle$ for each row $\textbf{g}_i$ ($i=1,2,\ldots,k$).
\end{itemize}
\label{R and RC RM code Gau}
\end{lemma}
\begin{proof}
For any code $\mathcal{R}(r,m,z)$ with the generator matrix $G_{r,m,z}$, the first row of $G_{r,m,z}$ is all one string, $i.e.$, $\textbf{g}_1$ = $\textbf{1}_{1,2^m}$, and therefore, the string $\textbf{1}_{1,2^m}\in\langle G_{r,m,z}\rangle$.
Thus, from closure property, $(2+2u)\textbf{1}_{1,2^m}\in\langle G_{r,m,z}\rangle$, and thus, $\textbf{\textit{2+2u}}_{1,2^m}\in\langle G_{r,m,z}\rangle$.
It follows the first part of the result.
From symmetry of the matrix $G_{r,m,z}$, it is easy to observe that, for each row $\textbf{g}_i$ ($i=1,2,\ldots,k$) if the matrix $G_{r,m,z}$, the reverse $\textbf{g}_i^r$ belongs to $\langle G_{r,m,z}\rangle$. 
Hence, it follows the result.
\end{proof}

Now, the properties of the DNA code obtained from the $r^{th}$ order Reed Muller type code $\mathcal{R}(r,m,z)$ is given in Theorem \ref{gen_r_rm gau}. 
\begin{theorem}
For any $(n,M,d_H)$ DNA code $\varphi_G(\mathcal{R}(r,m,z))$,  
\begin{itemize}
    \item Length: 
    \[
    n = 2^{m+1}
    \]
    \item Size: 
    \[
M = \left \{ \,
\begin{array}{ll}
2^{\left(4\sum_{i=0}^r\binom{m}{i}-3\sum_{i=0}^{r-1}\binom{m-1}{i}\right)} & \mbox{ if } z \in \{2u\},\\
2^{\left(4\sum_{i=0}^r\binom{m}{i}-2\sum_{i=0}^{r-1}\binom{m-1}{i}\right)} & \mbox{ if } z \in \{2,2+2u\}, \\
2^{\left(4\sum_{i=0}^r\binom{m}{i}-\sum_{i=0}^{r-1}\binom{m-1}{i}\right)} & \mbox{ if }  z\in \{u,2+u,3u,2+3u\}, \\
2^{\left(4\sum_{i=0}^r\binom{m}{i}\right)} & \mbox{ if }  z\mbox{ is a unit element of the ring }R, 
\end{array}
\right.\\
\]
    \item Minimum Hamming distance: 
    \[
d_H = \left \{ \,
\begin{array}{ll}
2^{m-r+1} & \mbox{ if } z \in \{2,2u,2+2u\},\\
2^{m-r} & \mbox{ if }  z\in R\backslash\{0,2,2u,2+2u\}.
\end{array}
\right.
\]
\end{itemize}
Further, the DNA code $\varphi_G(\mathcal{R}(r,m,z))$ is closed with R and RC DNA strings. 
\label{gen_r_rm gau}
\end{theorem}
\begin{proof}
The result on parameters of the DNA code $\varphi_G(\mathcal{R}(r,m,z))$ follows from Lemma \ref{Parameter RM code Gau} and Theorem \ref{distance preserving Gau}.
The result on reverse and reverse-complement properties follow from Lemma \ref{closure reverse Gau}, Lemma \ref{closure reverse comp Gau} and Lemma \ref{R and RC RM code Gau}. 
\end{proof}
From Theorem \ref{gen_r_rm gau}, the DNA code $\varphi_G(\mathcal{R}(r,m,z))$ satisfies
\begin{itemize}
\item Hamming constraint,
\item R constraint, and
\item RC constraint.
\end{itemize}


\subsection{DNA Codes from the Bijective Map over the Quinary Field}\label{sec: Z5}
For $t$ = $2$, consider $\mathscr{D}$ = $\{AA,AC,CA,CC,TC\}\subset\Sigma_{DNA}^2$ and $\mathcal{A}_q$ = $\mathbb{Z}_5$. Then, the map as given in (\ref{Bijective Map}) and the distance as shown in (\ref{Bijective hamming distance}) are the map and the distance discussed in \cite{9214530}.
We denote the set $\{AA,AC,CA,CC,TC\}$ by $\Sigma$ in Section \ref{sec: Z5}.

\subsubsection{The Bijective Map}
Consider a bijective map $\varphi:\mathbb{Z}_5\rightarrow\Sigma$ such that Table \ref{bijective Map on Z5} holds.
\begin{table}
\caption{The Bijective Map.}
\centering
\begin{tabular}{|p{2.5cm}||p{1cm}|p{1cm}|p{1cm}|p{1cm}|p{1cm}|}
\hline
Field element $x$    & $0$  & $1$    & $2$    & $3$    & $4$   \\ \hline
DNA image $\varphi(x)$ & $CC$ & $CA$   & $AC$   & $AA$   & $TC$   \\ \hline
\end{tabular}
\label{bijective Map on Z5}
\end{table}

For any $\x$ = $(x_1\ x_2\ \ldots\ x_n)\in\mathbb{Z}_5^n$, consider $\varphi(\x) = \varphi(x_1)\varphi(x_2)\ldots \varphi(x_n)\in\Sigma^n$. 
For any $\mathscr{C}\subseteq\mathbb{Z}_5^n$, $\varphi(\mathscr{C})=\{\varphi(\x):\x\in\mathscr{C}\}$.
Now, the properties of the map $\varphi$ as following.
\begin{lemma}
	Any DNA string defined over $\Sigma$ does not from any secondary structure with stems of length more than two. 
	\label{Sec Str property lemma Z5}
\end{lemma}
\begin{proof}
For $x_i\in\Sigma_{DNA}$ $(i=1,2,\ldots,2n)$ a DNA string $\x$ = $x_1x_2\ldots x_{2n}\in\Sigma^n$, consider a set $S_\x$ = $\{x_ix_{i+1}x_{i+2}:\mbox{ for }i=1,2,\ldots,2n-2\}\subseteq\Sigma_{DNA}^3$.
Then, for any $\x$ in $\Sigma^n$, $S_\x\subseteq \left\{AAA,AAC,ACA,CAA,CCA,CAC,ACC,CCC,TCA,TCC,TCT,ATC\right.$,  $\left.CTC,AAT,ACT,CAT,CCT,TCT\right\}$ = $S$. 
Now one can easily observe that, for any $z\in S$, $z^s$ and $z^{rs}$ are not belong to the set $S$.
Since any sub-string of length $3$ bps does not have its secondary-complement and reverse-secondary-complement DNA sub-strings in the DNA string, therefore, the DNA string is free from  secondary-complement and reverse-secondary-complement DNA sub-strings of length more than $2$ bps.
Thus, from Remark \ref{Sec Str existence remark}, the DNA string is independent from secondary structures of stem length more than two. 
\end{proof}
\begin{note}
In \cite{9214530}, authors have considered only reverse-secondary-complement DNA sub-strings to analysis secondary structures for any DNA string, and thus, in \cite[Lemma 3]{9214530}, they have concluded that any DNA string in $\Sigma^n$ is free from secondary structures of stem length more than one. 
\end{note}

\subsubsection{The Distance}
For any $x$ and $y$ in $\mathbb{Z}_5$, we define the distance 
\begin{equation}
    \begin{split}
        &d:\mathbb{Z}_5\times\mathbb{Z}_5\rightarrow\mathbb{R} \\ 
        &d(x,y) = H(\varphi(x),\varphi(y)).
    \end{split}
    \label{distance for Z5}
\end{equation}
Now, an isometry between $\mathbb{Z}_5^n$ and $\Sigma^n$ is given in Lemma \ref{bijective isometry preposition for Z5}.
\begin{lemma}
    The map $\varphi:(\mathbb{Z}_5^n,d) \rightarrow (\Sigma^n, d_H)$ is an isometry. 
\label{bijective isometry preposition for Z5}
\end{lemma} 
\begin{proof}
The result follows from Lemma \ref{bijective isometry preposition}. 
\end{proof}
From Lemma \ref{bijective isometry preposition for Z5}, one can calculate the parameters of constricted DNA codes as given in Theorem \ref{parameter map DNA code for Z5}. 
\begin{theorem}
if $\mathscr{C}$ is an $(n,M,d)$ code over $\mathbb{Z}_5$ then there exists a DNA code $\varphi(\mathscr{C})$ with the parameter $(2n,M,d_H)$, where $d$ = $d_H$. 
\label{parameter map DNA code for Z5}
\end{theorem}
\begin{proof}
The proof of the theorem follows from Theorem \ref{parameter map DNA code}.
\end{proof}
 A distance property on DNA strings defined over $\Sigma$ is given in Lemma \ref{complement lemma Z5}.
\begin{lemma}
	For any DNA strings $\x$ and $\y$ each of length $n$ defined over $\Sigma$, the Hamming distance $H(\x,\y^c)\geq n$. 
	\label{complement lemma Z5}
\end{lemma}
\begin{proof}
For any $x,y\in\Sigma$, note the Hamming distance $H(x,y^c)\geq1$.
Therefore,
\begin{equation*}
    \begin{split}
        H(\x,\y^c)=&\sum_{i=1}^nH(x_i,y_i^c) \\ 
        \geq& n.
    \end{split}
\end{equation*}
Hence, it follows the result.
\end{proof}
Now, an instant result on distance of obtained DNA codes using the Lemma \ref{complement lemma Z5} as given in Lemma \ref{rc constraint lemma Z5}. 
\begin{lemma}
	For any ($n,M,d$) code $\mathscr{C}$ over $\mathbb{Z}_5$, if the minimum distance $d\leq n$ then the DNA code $\varphi(\mathscr{C})$ satisfies RC constraint, . 
	\label{rc constraint lemma Z5}
\end{lemma}
\begin{proof}
Note that $(AA)^{rc}$ = $TT$, $(AC)^{rc}$ = $GT$, $(CA)^{rc}$ = $TG$, $(CC)^{rc}$ = $GG$ and $(TC)^{rc}$ = $GA$.
Thus, for any $x$ and $y$ in the set $\Sigma$, the minimum Hamming distance $H(x,y^{rc})\geq1$. 
Therefore, the minimum Hamming distance
\[
H(\x,\y^{rc})\geq n\geq d \mbox{ for }\x,\y\in\Sigma^n.
\]
Now, if $d\leq n$, then, from Lemma \ref{bijective isometry preposition for Z5}, for $\varphi(\Sigma^n)$, the minimum Hamming distance $d_H\leq n$, and therefore, $d_H\leq H(\x,\y^{rc})$ for each $\x,\y\in\Sigma^n$, where 
\[
d_H = \min\{H(\x,\y):\x,\y\in\varphi(\Sigma^n)\mbox{ and }\x\neq\y\}.
\]
Hence, for any $(2n,M,d_H^*)$ DNA code $\mathscr{C}_{DNA}\subseteq\varphi(\Sigma^n)$, $d_H^*\leq H(\x,\y^{rc})$ for each $\x,\y\in\mathscr{C}_{DNA}$, where 
\[
d_H^* = \min\{H(\x,\y):\x,\y\in\mathscr{C}_{DNA}\mbox{ and }\x\neq\y\}.
\]
It follows the result.
\end{proof}

\subsubsection{Constructions of DNA Codes}
For alphabet size five, from the family of linear codes constructed in \cite{BIER198747}, family of DNA codes are constructed in \cite{9214530}. 
For any integer $k$ = $2,3,4,5$, the generator matrix for the code is given by 
\[
G_k = \left(\begin{array}{ccccc}
\textbf{1}_{1,4^{k-2}} & \textbf{2}_{1,4^{k-2}} & \textbf{3}_{1,4^{k-2}} & \textbf{4}_{1,4^{k-2}} \\
G_{k-1} & G_{k-1} & G_{k-1} & G_{k-1} 
\end{array}\right) 
\mbox{ for }k=3,4,5,
\]
with the initial case
\[
G_2=\left(
\begin{array}{cccc}
1 & 1 & 1 & 1 \\
1 & 2 & 3 & 4 
\end{array}
\right).
\]
Using computation, one can easily obtain the Proposition \ref{Parameter code Z5} as follows.
\begin{proposition}
For $k=2,3,4,5$, if the code $\langle G_k\rangle$ is an $(n,M,d)$ code on $\mathbb{Z}_5$ then 
\begin{itemize}
    \item the length $n = 4^{k-1}$,
    \item the size $M = 5^k$, and
    \item the minimum distance $d = 3\cdot4^{k-2}$.
\end{itemize}
\label{Parameter code Z5}
\end{proposition}
Now, one can obtained the parameters of DNA codes as given in Theorem \ref{family 1 theorem Z5}.
\begin{theorem}
	For $k=2,3,4,5$, the DNA code $\varphi(\langle G_k\rangle)$ is $(2^{2k-1},5^k,3\cdot4^{k-2})$ code.
	\label{family 1 theorem Z5}
\end{theorem} 
\begin{proof}
The result can be obtained from Theorem \ref{parameter map DNA code for Z5} and Proposition \ref{Parameter code Z5}.
\end{proof}
From computation, one can obtain the result on the Hamming distance between DNA string and R DNA string as given in Proposition \ref{reverse code property Z5}. 
\begin{proposition}
For all $\x$ and $\y$ of the DNA code $\varphi(\langle G_k\rangle)$ $k=2,3,4,5$, the Hamming distance  $H(\x,\y^{rc})\geq2^{2k-3}$. 
\label{reverse code property Z5}
\end{proposition}
Now, the DNA code $\varphi(\langle G_k\rangle)$ $k=2,3,4,5$, 
\begin{itemize}
    \item has the parameters $(2^{2k-1},5^k,3\cdot4^{k-2})$ (from Theorem \ref{Parameter code Z5}), 
    \item satisfies the RC constraint (from Lemma \ref{rc constraint lemma Z5}), and
    \item $H(\x,\y^{rc})\geq2^{2k-3}$ for each $\x,\y\in\langle G_k\rangle$ (from Proposition \ref{reverse code property Z5}). 
\end{itemize}
Further, 
\begin{itemize}
    \item all DNA codewords of the DNA code $\varphi(\langle G_k\rangle)$ ($k=2,3,4,5$) are independent to the secondary structures of stem length two (from Lemma \ref{Sec Str property lemma Z5}), and
    \item DNA strings obtained from concatenation of codewords of the DNA code $\varphi(\langle G_k\rangle)$ is also independent to the secondary structures of stem length two.
\end{itemize}

\section{The Non-Homopolymer Map}\label{my sec:5}
In this section, we have established Non-Homopolymer map and distance.
And also studied their properties in this section.
Further, we have obtained DNA codes those are tandem-free and satisfy $GC$-content, R and RC constraints. 

\subsection{DNA Codes from the Non-Homopolymer Map}
$\ell$ order Non-Homopolymer map:
	For given any integer $\ell$ ($\geq1$) and $\x,\y\in\Sigma_{DNA}^\ell$ such that $\x\neq\y$, consider $\mathscr{S}$ = $\{\x,\y,\x^c,\y^c\}$. 
	Now, define a map 
	\[
	\psi:\mathbb{Z}_2\times\mathscr{S}\rightarrow\mathscr{S}
	\]
	such that 
	\begin{equation*}
	    \begin{array}{llll}
	         \psi(0,\x) = \y, & \psi(0,\x^c) = \y^c, & \psi(0,\y) = \x^c, & \psi(0,\y^c) = \x, \\ 
	         \psi(1,\x) = \y^c, & \psi(1,\x^c) = \y, & \psi(1,\y) = \x, & \psi(1,\y^c) = \x^c.
	    \end{array}
	\end{equation*}
For any $\textbf{a}$ = $(a_1\ a_2\ \ldots\ a_n)\in\mathbb{Z}_2^n$, consider 
\begin{equation}
    \begin{split}
        & \psi(\textbf{a}) = f(a_1)\psi(a_2,f(a_1))\psi(a_3,\psi(a_2,f(a_1)))\ldots \\ &\hspace{3cm}\ldots\psi(a_n,\psi(a_{n-1}\ldots\psi(a_2,f(a_1)\ldots)))\in\mathscr{S}^n
    \end{split}
    \label{conflict free map function}
\end{equation}
where $f:\mathbb{Z}_2\rightarrow\{\x^c,\x\}$ such that $f(0)$ = $\x$ and $f(1)$ = $\x^c$.
Again, for any $\textbf{a}$ = $(a_1\ a_2\ \ldots\ a_n)\in\mathbb{Z}_2^n$, if $\psi(\a)$ = $u_1u_2\ldots u_n$ in $\mathscr{S}^n$ then, using recurrence, 
\begin{equation*}
u_i  = 
    \begin{cases}
    \psi(a_i,u_{i-1}) & \mbox{ for }i=2,3,\ldots,n\mbox{ and } \\
        f(a_1) & \mbox{ for } i=1.
    \end{cases}
\end{equation*}
Now, for any $\mathscr{C}\subseteq\mathbb{Z}_2^n$, $\psi(\mathscr{C})=\{\psi(\x):\x\in\mathscr{C}\}$.
\begin{example}{Example}
If $\x$ = $ATA$ and $\y$ = $CGC$ then the binary string $(0\ 0\ 0\ 0)$ is encoded into a DNA string such that 
\begin{equation*}
\begin{array}{ccllll}
    \psi((0\ 0\ 0\ 0)) & = & f(0) & \psi(0,f(0)) & \psi(0,\psi(0,f(0))) & \psi(0,\psi(0,\psi(0,f(0)))) \\
    & = & \x& \psi(0,\x)& \psi(0,\psi(0,\x))& \psi(0,\psi(0,\psi(0,\x))) \\
    &= & \x& \y& \psi(0,\y)& \psi(0,\psi(0,\y)) \\
    &= & \x& \y& \x^c& \psi(0,\x^c) \\
    &= & \x& \y& \x^c& \y^c \\
    &= & ATA & CGC & TAT & GCG 
\end{array}
\end{equation*}
Thus, $\psi((0\ 0\ 0\ 0))$ = $ATACGCTATGCG$.
Again, for $\a$ = $(0\ 0\ 0\ 0)$, observe $u_1$ = $f(0)$ = $\x$, $u_2$ = $\psi(0,u_1)$ = $\y$,  $u_3$ = $\psi(0,u_2)$ = $\x^c$ and  $u_4$ = $\psi(0,u_3)$ = $\y^c$. 
Therefore, $\psi((0\ 0\ 0\ 0))$ = $u_1u_2u_3u_4$ = $\x \y \x^c \y^c$.
Similarly, 
\[
\begin{array}{cll}
    \psi((0\ 0\ 1\ 1)) & =\x\y\x\y^c & = ATACGCATAGCG, \\
    \psi((1\ 1\ 0\ 0)) & =\x^c\y\x^c\y^c & = TATCGCTATGCG, \mbox{ and }\\
    \psi((1\ 1\ 1\ 1)) & =\x^c\y\x\y^c & = TATCGCATAGCG.
\end{array}
\]
Thus, the binary code 
\[
\mathscr{C} = \{(0\ 0\ 0\ 0), (0\ 0\ 1\ 1), (1\ 1\ 0\ 0), (1\ 1\ 1\ 1)\}
\] 
is encoded into the $(12,4,3)$ DNA code 
\begin{equation*}
    \begin{split}
        &\left\{ATACGCTATGCG,ATACGCATAGCG,TATCGCTATGCG,\right. \\ 
        &\hspace{7cm}\left. TATCGCATAGCG\right\}.
    \end{split}
\end{equation*}
Observe that the binary code $\mathscr{C}$ is a linear code with the generator matrix
\[
G = 
\begin{pmatrix}
0 & 0 & 1 & 1 \\
1 & 1 & 0 & 0
\end{pmatrix}.
\]
	\label{non-homopolymer example}
\end{example}
For any tandem-free DNA string, the properties of the reverse, the complement and the RC DNA strings are given in Proposition \ref{reverse tandem free property}, Proposition \ref{complement tandem free property} and Proposition \ref{RC tandem free property} as follows.
\begin{proposition}
    A DNA string $\x$ is tandem-free DNA string with repeat-length $\ell$ if and only if $\x^r$ is tandem-free DNA string with repeat-length $\ell$.
    \label{reverse tandem free property}
\end{proposition}
\begin{proposition}
A DNA string $\x$ is tandem-free DNA string with repeat-length $\ell$ if and only if $\x^c$ is tandem-free DNA string with repeat-length $\ell$.
\label{complement tandem free property}
\end{proposition}
\begin{proposition}
A DNA string $\x$ is tandem-free DNA string with repeat-length $\ell$ if and only if $\x^{rc}$ is tandem-free DNA string with repeat-length $\ell$.
\label{RC tandem free property}
\end{proposition}
\begin{example}{Example}
For the tandem-free DNA string $ATACGCTATGCG$ with repeat-length $6$,
\begin{itemize}
    \item the R DNA string $GCGTATCGCATA$ is the tandem-free DNA string with repeat-length $6$, 
    \item the complement DNA string $TATGCGATACGC$ is the tandem-free DNA string with repeat-length $6$, and
    \item the RC DNA string $CGCATAGCGTAT$ is the tandem-free DNA string with repeat-length $6$. 
\end{itemize}
\end{example}
In Lemma \ref{tandem free lemma}, a property on a tandem-free DNA string is discussed that helps to ensure the property in DNA strings with larger length. 
\begin{lemma}
For any integers $\ell$ and $n$ ($2\ell\leq n$) and some $\x,\y\in\Sigma_{DNA}^\ell$, any binary string of length $n$ will encode into a tandem-free DNA string with repeat-length $\ell$ using the $\ell$ order Non-Homopolymer map, if the DNA strings $\x\y$, $\x\y^c$, $\y\x$ and $\y\x^c$ are also tandem-free DNA strings with repeat-length $\ell$.
\label{tandem free lemma}
\end{lemma}
\begin{proof}
For given $\x,\y\in\Sigma_{DNA}^\ell$ and $\mathscr{S}$ = $\{\x,\y,\x^c,\y^c\}$, if the DNA strings $\x\y$, $\x\y^c$, $\y\x$ and $\y\x^c$ are all tandem-free DNA string with repeat-length $\ell$ then, from Proposition \ref{complement tandem free property}, all the DNA strings in the set $A$ = $\{\x\y,\x\y^c, \x^c\y, \x^c\y^c, \y\x, \y\x^c, \y^c\x, \y^c\x^c\}$ are tandem-free DNA string with repeat-length $\ell$. 
Thus, for any binary string $\a$ = $(a_1\ a_2\ \ldots\ a_n)\in\mathbb{Z}_2^n$, consider the encoded DNA string $\psi(\a)$ = $\u$ = $u_1u_2\ldots u_n$ in $\mathscr{S}^n$ that is obtained using $\ell$ order Non-Homopolymer map on $\a$, where
\begin{equation*}
u_i  = 
    \begin{cases}
    \psi(a_i,u_{i-1}) & \mbox{ for }i=2,3,\ldots,n\mbox{ and } \\
        f(a_1) & \mbox{ for } i=1.
    \end{cases}
\end{equation*}
Now, for $i=1,2,\ldots,n-1$, $u_iu_{i+1}\in A$, and thus, $u_iu_{i+1}$ is tandem-free DNA string with repeat-length $\ell$ for each $i$.
Hence, the encoded DNA string is tandem-free DNA string with repeat-length $\ell$.
\end{proof}
The $GC$-weight of the DNA string that is obtained from Homopolymer map applied on any binary string is discussed in Lemma \ref{kg gen gc cont}. \begin{lemma}
	For any integers $\ell$ ($\geq1$) and $n$ ($\geq1$), and given DNA strings $\x,\y\in\Sigma_{DNA}^\ell$, the $GC$-weight of any DNA string  $\u\in\psi(\mathbb{Z}_2^n)$ is
	\begin{equation*}
	w_{GC}(\u) = 
	\left\{
	\begin{array}{ll}
	\frac{n}{2}(w_{GC}(u_1)+w_{GC}(u_2)) & \mbox{ if } n \mbox{ is even integer,} \\
	w_{GC}(u_1)+\frac{(n-1)}{2}\left(w_{GC}(u_1)+w_{GC}(u_2)\right) & \mbox{ if } n \mbox{ is odd integer.}
	\end{array}
	\right.
	\end{equation*}
	\label{kg gen gc cont}
\end{lemma}
\begin{proof} 
	For any integers $\ell$ ($\geq1$) and $n$ ($\geq1$), if a binary string $\a$ = $(a_1\ a_2 \ldots a_n)\in\mathbb{Z}_2^n$ is encoded into the DNA string $\u$ = $u_1u_2\ldots u_n\in\psi(\mathbb{Z}_2^n)$ using the $\ell$ order Non-Homopolymer map. Now, the $GC$-weight 
	\begin{equation*}
	    \begin{split}
	        w_{GC}(\u) & = \sum_{i=1}^nw_{GC}(u_i) \\ 
	                   & = 
	                   \begin{cases}
	                   \sum_{j=1}^{n/2}(w_{GC}(u_{2j-1})+w_{GC}(u_{2j})) & \mbox{ if }n\mbox{ is even} \\
	                   w_{GC}(u_1)+\sum_{j=1}^{(n-1)/2}(w_{GC}(u_{2j})+w_{GC}(u_{2j+1})) & \mbox{ if }n\mbox{ is odd}
	                   \end{cases}
	    \end{split}
	\end{equation*}
	But, from $\ell$ order Non-Homopolymer map, the $GC$-weight 
	\[
	w_{GC}(u_i)=w_{GC}(u_{i+2})\mbox{ for }i=1,2,\ldots,n-2.
	\]
	Therefore, 
	\begin{equation*}
	    \begin{split}
	        & w_{GC}(u_{2j-1})+w_{GC}(u_{2j}) = w_{GC}(u_1)+w_{GC}(u_2) \mbox{ for }j=1,2,\ldots,n/2, \mbox{ and } \\
	        & w_{GC}(u_{2j})+w_{GC}(u_{2j+1}) = w_{GC}(u_2)+w_{GC}(u_3)\mbox{ for }j=1,2,\ldots,(n-1)/2. 
	    \end{split}
	\end{equation*} 
	Also, $w_{GC}(u_1)$ = $w_{GC}(u_3)$.
	Thus,
	\begin{equation*}
	    \begin{split}
	        w_{GC}(\u) & = 
	                   \begin{cases}
	                   \sum_{j=1}^{n/2}(w_{GC}(u_1)+w_{GC}(u_2)) & \mbox{ if }n\mbox{ is even} \\
	                   w_{GC}(u_1)+\sum_{j=1}^{(n-1)/2}(w_{GC}(u_2)+w_{GC}(u_3)) & \mbox{ if }n\mbox{ is odd}
	                   \end{cases} \\
	                   & = 
	                   \begin{cases}
	                   \sum_{j=1}^{n/2}(w_{GC}(u_1)+w_{GC}(u_2)) & \mbox{ if }n\mbox{ is even} \\
	                   w_{GC}(u_1)+\sum_{j=1}^{(n-1)/2}(w_{GC}(u_1)+w_{GC}(u_2)) & \mbox{ if }n\mbox{ is odd}
	                   \end{cases} \\
	                   & = 
	                   \begin{cases}
	                   \frac{n}{2}(w_{GC}(u_1)+w_{GC}(u_2)) & \mbox{ if }n\mbox{ is even} \\
	                   w_{GC}(u_1)+\frac{(n-1)}{2}(w_{GC}(u_1)+w_{GC}(u_2)) & \mbox{ if }n\mbox{ is odd}.
	                   \end{cases}
	    \end{split}
	\end{equation*} 
	It follows the result. 
\end{proof}
From Lemma \ref{kg gen gc cont}, one can obtain Proposition \ref{kg gen gc cont proposition}, and further, Proposition \ref{KG GC cont} that ensures the $GC$-weight for encoded DNA codes.
\begin{proposition}
For any integers $\ell$ ($\geq1$) and $n$ ($\geq1$), and given DNA strings $\x,\y\in\Sigma_{DNA}^\ell$, if $w_{GC}(\x)+w_{GC}(\y)$ = $\ell$ then the $GC$-weight of any DNA string $\u\in\psi(\mathbb{Z}_2^n)$ is
	\begin{equation*}
	w_{GC}(\u) = 
	\left\{
	\begin{array}{ll}
	w_{GC}(u_1)+\frac{(n-1)}{2}\ell & \mbox{ if } n \mbox{ is odd integer} \\
	\frac{n}{2}\ell & \mbox{ if } n \mbox{ is even integer}.
	\end{array}
	\right.
	\end{equation*}
	\label{kg gen gc cont proposition}
\end{proposition}

\begin{example}{Example}
For $\ell$ = $2$, if $\x$ = $AT$ and $\y$ = $CG$ then $w_{GC}(\psi(\x))$ = $0$, and $w_{GC}(\psi(\y))$ = $2$. 
\begin{itemize}
    \item Now, for $n$ = $3$ (a odd integer), if $\a\in\mathbb{Z}_2^3$ then $\psi(\a)$ = $\u$ = $u_1u_2u_3$, where $u_1,u_3\in\{AT,TA\}$ and $u_2\in\{GC,CG\}$. 
    Therefore, $w_{GC}(u_1)$ = $w_{GC}(\psi(\x))$ = $0$.
    In this case, from Proposition \ref{kg gen gc cont proposition}, $w_{GC}(\psi(\a))$ = $0+\frac{(3-1)}{2}\cdot2$ = $2$, and it can be verified as follows.
\[
      \begin{array}{c||c|c|c}\hline
          \a        & \psi(\a)     & \u     & w_{GC}(\u) \\ \hline
          (0\ 0\ 0) & \x\y\x^c     & ATCGTA & 2 \\
          (0\ 0\ 1) & \x\y\x       & ATCGAT & 2 \\
          (0\ 1\ 0) & \x\y^c\x     & ATGCAT & 2 \\
          (0\ 1\ 1) & \x\y^c\x^c   & ATGCTA & 2 \\
          (1\ 0\ 0) & \x^c\y^c\x   & TAGCAT & 2 \\
          (1\ 0\ 1) & \x^c\y^c\x^c & TAGCTA & 2 \\
          (1\ 1\ 0) & \x^c\y\x^c   & TACGTA & 2 \\
          (1\ 1\ 1) & \x^c\y\x     & TACGAT & 2 \\ \hline
      \end{array}
\]
    \item Also, for $n$ = $4$ (an even integer), if $\a\in\mathbb{Z}_2^4$ then $\psi(\a)$ = $\u$ = $u_1u_2u_3u_4$, where $u_1,u_3\in\{AT,TA\}$ and $u_2,u_4\in\{GC,CG\}$. 
    Therefore, $w_{GC}(u_1)$ = $w_{GC}(\psi(\x))$ = $0$.
    Again, from Proposition \ref{kg gen gc cont proposition}, $w_{GC}(\psi(\a))$ = $0+\frac{4}{2}\cdot2$ = $4$, and it can be verified as follows.
    \[
\begin{array}{c||c|c|c}
           \a           & \psi(\a)            & \u       & w_{GC}(\u) \\ \hline
           (0\ 0\ 0\ 0) & \x \y \x^c \y^c     & ATCGTAGC & 4 \\
           (0\ 0\ 0\ 1) & \x \y \x^c \y       & ATCGTACG & 4 \\
           (0\ 0\ 1\ 0) & \x \y \x \y         & ATCGATCG & 4 \\
           (0\ 0\ 1\ 1) & \x \y \x \y^c       & ATCGATGC & 4 \\
           (0\ 1\ 0\ 0) & \x \y^c \x \y       & ATGCATCG & 4 \\
           (0\ 1\ 0\ 1) & \x \y^c \x \y^c     & ATGCATGC & 4 \\
           (0\ 1\ 1\ 0) & \x \y^c \x^c \y^c   & ATGCTAGC & 4 \\
           (0\ 1\ 1\ 1) & \x \y^c \x^c \y     & ATGCTACG & 4 \\
           (1\ 0\ 0\ 0) & \x^c \y^c \x \y     & TAGCATCG & 4 \\
           (1\ 0\ 0\ 1) & \x^c \y^c \x \y^c   & TAGCATGC & 4 \\
           (1\ 0\ 1\ 0) & \x^c \y^c \x^c \y^c & TAGCTAGC & 4 \\
           (1\ 0\ 1\ 1) & \x^c \y^c \x^c \y   & TAGCTACG & 4 \\
           (1\ 1\ 0\ 0) & \x^c \y \x^c \y^c   & TACGTAGC & 4 \\
           (1\ 1\ 0\ 1) & \x^c \y \x^c \y     & TACGTACG & 4 \\
           (1\ 1\ 1\ 0) & \x^c \y \x \y       & TACGATCG & 4 \\
           (1\ 1\ 1\ 1) & \x^c \y \x \y^c     & TACGATGC & 4 
      \end{array}
\]
\end{itemize}
\end{example}

\begin{proposition}
For any integers $\ell$ ($\geq1$) and $n$ ($\geq1$), and given DNA strings $\x,\y\in\Sigma_{DNA}^\ell$, the $GC$-weight of any DNA string  $\u\in\psi(\mathbb{Z}_2^n)$ is
\begin{equation*}
	w_{GC}(\u) = 
	\left\{
	\begin{array}{ll}
	\lfloor n\ell/2\rfloor & \mbox{ if } w_{GC}(\x)=\lfloor\ell/2\rfloor\mbox{ and } w_{GC}(\y)=\lceil\ell/2\rceil,  \\
	\lceil n\ell/2 \rceil & \mbox{ if } w_{GC}(\x)=\lceil\ell/2\rceil\mbox{ and } w_{GC}(\y)=\lfloor\ell/2\rfloor.
	\end{array}\right.
	\end{equation*}
	\label{KG GC cont}
\end{proposition}

\begin{example}{Example}
For $\ell$ = $3$ and $n$ = $3$, if $\x,\y\in\Sigma_{DNA}^3$ and $\a\in\mathbb{Z}_2^3$ then consider $\psi(\a)$ = $\u$ = $u_1u_2u_3$, where $u_1,u_3\in\{\x^c,\x\}$ and $u_2\in\{\y,^c,\y\}$.
Then one can observe the following.
\begin{itemize}
    \item If $\x$ = $ACA$ and $\y$ = $CTC$ then $w_{GC}(\psi(\x))$ = $\lfloor 3/2\rfloor$ = $1$, and $w_{GC}(\psi(\y))$ = $\lceil 3/2\rceil$ = $2$.
    In this case, from Proposition \ref{KG GC cont}, $w_{GC}(\u)$ = $\lfloor 3\cdot3/2\rfloor$ = $4$, and it can be verified as follows.
\[
      \begin{array}{c||c|c|c}\hline
          \a        & \psi(\a)     & \u     & w_{GC}(\u) \\ \hline
          (0\ 0\ 0) & \x\y\x^c     & ACACTCTGT & 4 \\
          (0\ 0\ 1) & \x\y\x       & ACACTCACA & 4 \\
          (0\ 1\ 0) & \x\y^c\x     & ACAGAGACA & 4 \\
          (0\ 1\ 1) & \x\y^c\x^c   & ACAGAGTGT & 4 \\
          (1\ 0\ 0) & \x^c\y^c\x   & TGTGAGACA & 4 \\
          (1\ 0\ 1) & \x^c\y^c\x^c & TGTGAGTGT & 4 \\
          (1\ 1\ 0) & \x^c\y\x^c   & TGTCTCTGT & 4 \\
          (1\ 1\ 1) & \x^c\y\x     & TGTCTCACA & 4 \\ \hline
      \end{array}
\]
    \item Also, if $\x$ = $CGA$ and $\y$ = $CAT$ then $w_{GC}(\psi(\x))$ = $\lceil 3/2\rceil$ = $2$, and $w_{GC}(\psi(\y))$ = $\lfloor 3/2\rfloor$ = $1$.
    In this case, from Proposition \ref{KG GC cont}, $w_{GC}(\u)$ = $\lceil 3\cdot3/2\rceil$ = $5$, and it can be verified as follows.
    \[
\begin{array}{c||c|c|c}
           \a       & \psi(\a)     & \u        & w_{GC}(\u) \\ \hline
          (0\ 0\ 0) & \x\y\x^c     & CGACATGCT & 5 \\
          (0\ 0\ 1) & \x\y\x       & CGACATCGA & 5 \\
          (0\ 1\ 0) & \x\y^c\x     & CGAGTACGA & 5 \\
          (0\ 1\ 1) & \x\y^c\x^c   & CGAGTAGCT & 5 \\
          (1\ 0\ 0) & \x^c\y^c\x   & GCTGTACGA & 5 \\
          (1\ 0\ 1) & \x^c\y^c\x^c & GCTGTAGCT & 5 \\
          (1\ 1\ 0) & \x^c\y\x^c   & GCTCATGCT & 5 \\
          (1\ 1\ 1) & \x^c\y\x     & GCTCATCGA & 5 \\ \hline
      \end{array}
\]
\end{itemize}
\end{example}

\subsection{The Non-Homopolymer Distance and Properties}
Now, we define a distance as given in Definition \ref{NHo distance def} for any alphabet of size $q$ such that the distance is equal to the Hamming distance in the respective DNA codes for a special case of binary alphabet.
\begin{definition}
	For any integer $n$ ($> 1$) and an alphabet $\mathcal{A}_q$ ($q\leq2$), consider $\a$ = $(a_1\ a_2\ldots a_n)$ and $\b$ = $(b_1\ b_2\ldots b_n)$ in $\mathcal{A}_q^n$.
	Now, for the support set 
	\[
	S = \{i:i=1,2,\ldots,n\mbox{ and }a_i\neq b_i\},
	\]
	and the set 
	\begin{equation*}
	T =
	\left\{
	\begin{array}{ll}
	S\cup\{n+1\} & \mbox{ if the size of the set } S \mbox{ is odd}, \\
	S & \mbox{ if the size of the set } S \mbox{ is even}, 
	\end{array}
	\right.
	\end{equation*}
	if the extended support set $T$ is a nonempty set then consider $T$ = $\{t_1,t_2,\ldots,t_{|T|}\}$ such that $t_j<t_{j+1}$ for $j=1,2,\ldots,|T|-1$, where $|T|$ represents the size of the set $T$. 
	For any integer $\ell$ ($\geq1$), define a map 
	\begin{equation*}
	\begin{split}
	    & d_{NHo}:\mathcal{A}_q^n\times\mathcal{A}_q^n\rightarrow\mathbb{R} \mbox{ such that }\\
	    & d_{NHo}(\a,\b) = 
	\left\{
	\begin{array}{ll}
	\ell\sum_{j=1}^{|T|/2} (t_{2j}-t_{2j-1}) & \mbox{ if } |T|>0,\\
	0 & \mbox{ if } |T|=0.
	\end{array}
	\right.
	\end{split}
	\end{equation*}
	\label{NHo distance def}
\end{definition}
\begin{example}{Example}
For $n=5$ and $\ell=3$, consider $\a=(1\ 0\ 0\ 0\ 0)$ and $\b=(1\ 1\ 1\ 0\ 1)$ in $\mathbb{Z}_2^5$.
Then the support set $S$ = $\{2,3,5\}$, and thus, the extended support set $T$ = $\{2,3,5,6\}$.
Therefore, 
\begin{equation*}
    \begin{split}
        d_{NHo}(\a,\b) & =  3\left((3-2)+(6-5)\right) \\
        & = 6.
    \end{split}
\end{equation*}
\end{example}
From Definition \ref{NHo distance def}, one can observe Remark \ref{distance remark 1} and Remark \ref{distance remark 2} as follows.
\begin{remark}
For $\x,\y\in\Sigma_{DNA}^\ell$ and any $\a\in\mathbb{Z}_2^n$, if $\psi(\a)$ = $\u$ = $u_1u_2\ldots u_n$ in $\psi(\mathbb{Z}_2^n)$ then 
\[
u_i\in
\begin{cases}
\{\x^c,\x\} & \mbox{ if }i\mbox{ is odd, and} \\
\{\y,^c,\y\} & \mbox{ if }i\mbox{ is even.}
\end{cases}
\]
\label{distance remark 1}
\end{remark}
\begin{remark}
For $\x,\y\in\Sigma_{DNA}^\ell$ and any $\a,\b\in\mathbb{Z}_2^n$, consider $\psi(\a)$ = $\u$ = $ u_1 u_2\ldots u_n$ and $\psi(\b)$ = $\v$ = $ v_1 v_2\ldots v_n$ in $\psi(\mathbb{Z}_2^n)$ with support set $S$ = $\{t_1,t_2,\ldots,t_s\}$ of size $s$ such that $1\leq t_1<t_2<\ldots<t_s\leq n$.
Then, 
\begin{itemize}
    \item if $t_1>1$ then the DNA sub-strings $ u_1 u_2\ldots u_{t_1-1}$ and $ v_1 v_2\ldots v_{t_1-1}$ exist, and 
    \[
     u_1 u_2\ldots u_{t_1-1} =  v_1 v_2\ldots v_{t_1-1},
    \]
    \item for any odd integer $i$, if $t_i$ and $t_{i+1}$ are in the extended support set $T$ then the DNA sub-strings $ u_{t_i} u_{t_i+1}\ldots u_{t_{i+1}-1}$ and $ v_{t_i} v_{t_i+1}\ldots v_{t_{i+1}-1}$ exist, and
    \[
     u_{t_i} u_{t_i+1}\ldots u_{t_{i+1}-1} =  v_{t_i}^c v_{t_i+1}^c\ldots v_{t_{i+1}-1}^c,
    \]
    \item for any even integer $i$, if $t_i$ and $t_{i+1}$ are in the extended support set $T$ then the DNA sub-strings $ u_{t_i} u_{t_i+1}\ldots u_{t_{i+1}-1}$ and $ v_{t_i} v_{t_i+1}\ldots v_{t_{i+1}-1}$ exist, and
    \[
     u_{t_i} u_{t_i+1}\ldots u_{t_{i+1}-1} =  v_{t_i} v_{t_i+1}\ldots v_{t_{i+1}-1}, and
    \]
    \item if $t_s\leq n$ then the DNA sub-strings $ u_s u_{s+1}\ldots u_n$ and $ v_s v_{s+1}\ldots v_n$ exist, and 
    \[
     u_s u_{s+1}\ldots u_n = 
    \begin{cases}
     v_s v_{s+1}\ldots v_n & \mbox{ if }s\mbox{ is even, and} \\
     v_s^c v_{s+1}^c\ldots v_n^c & \mbox{ if }s\mbox{ is odd.}
    \end{cases}
    \]
\end{itemize}
\label{distance remark 2}
\end{remark}
We have shown in Lemma \ref{Nho distance proof} that the real map as given in Definition \ref{NHo distance def} is a distance.
For that first we need a result that is given in Lemma \ref{recurrence on Non Homopolymer distance}. 
\begin{lemma}
For any integer $\ell$ ($\geq1$), any $\a,\b\in\mathcal{A}_q^n$ and any $a,b\in\mathcal{A}_q$, 
\[
d_{NHo}((\a\ a),(\b\ b)) = 
\begin{cases}
d_{NHo}(\a,\b) & \mbox{ if }a=b\mbox{ and }|S|\mbox{ is even,} \\
\ell+d_{NHo}(\a,\b) & \mbox{ if }a=b\mbox{ and }|S|\mbox{ is odd,} \\
\ell+d_{NHo}(\a,\b) & \mbox{ if }a\neq b\mbox{ and }|S|\mbox{ is even, and} \\
d_{NHo}(\a,\b) & \mbox{ if }a\neq b\mbox{ and }|S|\mbox{ is odd.} 
\end{cases}
\]
\label{recurrence on Non Homopolymer distance}
\end{lemma}
\begin{proof}
For any $\a$ and $\b$ in $\mathcal{A}_q^n$, the support set and extended support set are $S$ and $T$.
For any $a,b\in\mathcal{A}_q$, consider $(\a\ a)$ and $(\b\ b)$ in $\mathcal{A}_q^{n+1}$ along with the support set $S^*$ and extended support set $T^*$. 
Then, from Definition \ref{NHo distance def}, the support set
\begin{equation*}
    S^* =
    \begin{cases}
    S & \mbox{ if }a=b, \\
    S\cup\{|S|+1\} & \mbox{ if }a\neq b,
    \end{cases}
\end{equation*}
and the extended support set 
\begin{equation*}
    T^* =
    \begin{cases}
S & \mbox{ if }a = b\mbox{ and }|S|\mbox{ is even,} \\
S\cup\{|S|+2\} & \mbox{ if }a = b\mbox{ and }|S|\mbox{ is odd,} \\
S\cup\{|S|+1,|S|+2\} & \mbox{ if }a\neq b\mbox{ and }|S|\mbox{ is even, and} \\
S\cup\{|S|+1\} & \mbox{ if }a\neq b\mbox{ and }|S|\mbox{ is odd.} 
    \end{cases}
\end{equation*}
Therefore, from Definition \ref{NHo distance def}, 
\[
d_{NHo}((\a\ a),(\b\ b)) = 
\begin{cases}
d_{NHo}(\a,\b) & \mbox{ if }a=b\mbox{ and }|S|\mbox{ is even,} \\
\ell+d_{NHo}(\a,\b) & \mbox{ if }a=b\mbox{ and }|S|\mbox{ is odd,} \\
\ell+d_{NHo}(\a,\b) & \mbox{ if }a\neq b\mbox{ and }|S|\mbox{ is even, and} \\
d_{NHo}(\a,\b) & \mbox{ if }a\neq b\mbox{ and }|S|\mbox{ is odd.} 
\end{cases}
\]
It follows the result.
\end{proof}
\begin{lemma}
The map $d_{NHo}:\mathcal{A}_q\times\mathcal{A}_q\rightarrow\mathbb{R}$, as given in Definition \ref{NHo distance def}, is a distance.
\label{Nho distance proof}
\end{lemma}
\begin{proof} 
A real map is called distance if the map follows non-negative property, identity of indiscernibles property, symmetry property and triangular property.
For the real map $d_{NHo}$, one can observe the following. 
\begin{description}[Identity Of Indiscernibles:]
    \item[Non-Negative Property:] For any integer $\ell$ ($\geq1$) and any $\a,\b\in\mathcal{A}_q^n$, consider the nonempty extended support set $T$ = $\{t_1,t_2,\ldots,t_{|T|}\}$, where $t_j<t_{j+1}$ for $j=1,2,\ldots,|T|-1$.
    Then, 
    \begin{equation*}
        \begin{split}
            & t_{2j}-t_{2j-1}>0 \mbox{ for }j=1,2,\ldots,|T|/2  \\
           \Rightarrow\hspace{0.1cm} & \ell\sum_{j=1}^{|T|/2} (t_{2j}-t_{2j-1})>0 \\
           \Rightarrow\hspace{0.1cm} & d_{NHo}(\a,\b)>0 \mbox{ for any }\a,\b\in\mathcal{A}_q^n. \\
        \end{split}
    \end{equation*}
    Now, if the empty extended support set is empty, $i.e.$, $T$ = $\emptyset$ then the proof for the non-negative property is trivial.
    \item[Identity of Indiscernibles:] For any $\a$ = $(a_1\ a_2\ \ldots\ a_n)$ and $\b$ = $(b_1\ b_2\ \ldots\ b_n)$ in $\mathcal{A}_q^n$, the distance 
    \begin{equation*}
        \begin{split}
            & d_{NHo}(\a,\b) = 0 \\
           \Leftrightarrow\hspace{0.1cm} & T = \emptyset \\
           \Leftrightarrow\hspace{0.1cm} & S = \emptyset \\
           \Leftrightarrow\hspace{0.1cm} & a_i = b_i \mbox{ for }i=1,2,\ldots,n \\
           \Leftrightarrow\hspace{0.1cm} & \a=\b.
        \end{split}
    \end{equation*}
    \item[Symmetry Property:] For any $\a,\b\in\mathcal{A}_q^n$, 
    the support set for the both $d_{NHo}(\a,\b)$ and $d_{NHo}(\b,\a)$ are the same, and thus, $d_{NHo}(\a,\b)$ = $d_{NHo}(\b,\a)$.
    \item[Triangular Property:] Using Mathematical Induction over $n$, we have shown the triangle property for $d_{NHo}$. \\
    Base Case: For $n=1$, it is easy to verify that the map $d_{NHo}$ holds Triangle property $d_{NHo}(a,b)\leq d_{NHo}(a,c)+d_{NHo}(c,b)$ for any $a,b,c\in\mathcal{A}_q$. \\
    Hypothesis: For $n=k$ and any $\a,\b,\c\in\mathcal{A}_q^k$, we assume that the map $d_{NHo}$ holds Triangle property, $i.e.$,  \[d_{NHo}(\a,\b)\leq d_{NHo}(\a,\c)+d_{NHo}(\c,\b).\] \\
    Inductive Step: For any $\a,\b,\c\in\mathcal{A}_q^k$ and any $a,b\in\mathcal{A}_q$, consider support sets $S_{\a,\b}$, $S_{\a,\c}$ and $S_{\c,\b}$ for $d_{NHo}(\a,\b)$, $d_{NHo}(\a,\c)$ and $d_{NHo}(\c,\b)$, respectively.
    Now, from Lemma \ref{recurrence on Non Homopolymer distance}, \\
    $d_{NHo}((\a\ a),(\b\ b))$ =
    \[
\begin{cases}
d_{NHo}(\a,\b) & \mbox{ if }a=b\mbox{ and }|S_{\a,\b}|\mbox{ is even,} \\
\ell+d_{NHo}(\a,\b) & \mbox{ if }a=b\mbox{ and }|S_{\a,\b}|\mbox{ is odd,} \\
\ell+d_{NHo}(\a,\b) & \mbox{ if }a\neq b\mbox{ and }|S_{\a,\b}|\mbox{ is even,} \\
d_{NHo}(\a,\b) & \mbox{ if }a\neq b\mbox{ and }|S_{\a,\b}|\mbox{ is odd,} 
\end{cases}
\]
$d_{NHo}((\a\ a),(\c\ c))$ =  
\[
\begin{cases}
d_{NHo}(\a,\c) & \mbox{ if }a=c\mbox{ and }|S_{\a,\c}|\mbox{ is even,} \\
\ell+d_{NHo}(\a,\c) & \mbox{ if }a=c\mbox{ and }|S_{\a,\c}|\mbox{ is odd,} \\
\ell+d_{NHo}(\a,\c) & \mbox{ if }a\neq c\mbox{ and }|S_{\a,\c}|\mbox{ is even,} \\
d_{NHo}(\a,\c) & \mbox{ if }a\neq c\mbox{ and }|S_{\a,\c}|\mbox{ is odd,} 
\end{cases}
\]
and $d_{NHo}((\c\ c),(\b\ b))$ =   
\[
\begin{cases}
d_{NHo}(\c,\b) & \mbox{ if }c=b\mbox{ and }|S_{\c,\b}|\mbox{ is even,} \\
\ell+d_{NHo}(\c,\b) & \mbox{ if }c=b\mbox{ and }|S_{\c,\b}|\mbox{ is odd,} \\
\ell+d_{NHo}(\c,\b) & \mbox{ if }c\neq b\mbox{ and }|S_{\c,\b}|\mbox{ is even, and} \\
d_{NHo}(\c,\b) & \mbox{ if }c\neq b\mbox{ and }|S_{\c,\b}|\mbox{ is odd.} 
\end{cases}
\]
Now, for various cases, one can easily obtain that 
\[
d_{NHo}(\a,\b)\leq d_{NHo}(\a,\c)+d_{NHo}(\c,\b).
\]
So, the map $d_{NHo}$ follows the triangle property for $n$ = $k+1$.
Thus, from Mathematical Induction, $d_{NHo}$ follows the Triangle property
\end{description}
Hence, from the distance definition, the map given in Definition \ref{NHo distance def} is a distance.
\end{proof}
For any code $\mathscr{C}\subseteq\mathcal{A}_q^n$, the minimum Non-Homopolymer distance is
\[
d_{NHo} = \min\{d_{NHo}(\a,\b):\a,\b\in\mathscr{C}\mbox{ and }\a\neq\b\}.
\]
In Remark \ref{d_Nho and Hamming distance bound}, we have obtained a bound on the minimum Non-Homopolymer distance as follows.
\begin{remark}
For any $\a,\b\in\mathcal{A}_q^n$, from Definition \ref{NHo distance def}, one can observe that the size of the support set is the Hamming distance $H(\a,\a)$.
Therefore, the Non-Homopolymer distance $d_{NHo}(\x,\y)\geq\lceil H(\x,\y)/2\rceil$, and thus, for any code with the minimum Non-Homopolymer distance $d_{NHo}$ and the minimum Hamming distance $d_H$, 
\[
\lceil d_H/2\rceil\leq d_{NHo}.
\]
\label{d_Nho and Hamming distance bound}
\end{remark}
Now, bounds on various Hamming distances are calculated in Theorem \ref{kg reverse condition} and Proposition \ref{kg reverse-complement condition} that helps to study the R and RC constraints in DNA codes obtained from binary codes.
\begin{theorem}
	For any given integers $\ell$ and $n$ ($\ell,n\geq1$), consider $\x,\y\in\Sigma_{DNA}^\ell$.
	Then, for DNA strings $\u,\v\in\psi(\mathbb{Z}_2^n)$, the Hamming distance 
	\begin{equation*}
	H(\u,\v^r) \geq  
	\left\{
	\begin{array}{ll}
	n\min\{H(\x,\y^r),H(\x,\y^{rc})\}, & \mbox{ if } n \mbox{ is even,} \\
	\min\left\{H(\x,\x^r),H(\y,\y^r),H(\x,\x^{rc}),H(\y,\y^{rc})\right\}, & \mbox{ if } n \mbox{ is odd. } 
	\end{array}\right.
	\end{equation*}
	\label{kg reverse condition}
\end{theorem}
\begin{proof}
	For $\x,\y\in\Sigma_{DNA}^\ell$, consider binary strings $\a,\b\in(\mathbb{Z}_2^n)$ of length $n$ and these strings are encoded into DNA strings $\psi(\a)$ = $\u$ = $u_1 u_2\ldots u_n$ and $\psi(\b)$ = $\v$ = $v_1 v_2\ldots v_n$ in $\psi(\mathbb{Z}_2^n)$ using $\ell$ order Non-Homopolymer Map,	where $u_{2i},v_{2i}\in\{\y^c,\y\}$ and $u_{2i-1},v_{2i-1}\in\{\x^c,\x\}$ for $i=1,2,\ldots,n$. 
	Consider \[H(\u,\v^r)=\sum_{j=1}^nH(u_j,v^r_{n-j+1}).\] 
	Now, there are two cases as follows.
	\begin{description}[\hspace{1.2cm}]
	    \item[Odd $n$:] In this case, $j$ is even (odd) if and only if $n-j+1$ is even (odd).
	    Thus, if $j$ is even then $u_j,v_{n-j+1}\in\{\y^c,\y\}$, and if $j$ is odd then $u_j,v_{n-j+1}\in\{\x^c,\x\}$.
	    So, \[H(u_j,v^r_{n-j+1})\geq\min\{H(\x,\x^r),H(\y,\y^r),H(\x,\x^{rc}), H(\y,\y^{rc})\}.\]
	    \item[Even $n$:] In this case, $j$ is even (odd) if and only if $n-j+1$ is odd (even).
	    Thus, if $j$ is even then $u_j\in\{\y,^c,\y\}$ and $v_{n-j+1}\in\{\x^c,\x\}$.
	    And, if $j$ is odd then $u_j\in\{\x^c,\x\}$ and $v_{n-j+1}\in\{\y^c,\y\}$.
	    Thus, \[H( u_j, v_{n-j+1}^r)\geq\min\{H(\x,\y^r),H(\x,\y^{rc})\}.\]
	\end{description}
	Hence, the result follows for any integer $n$.
\end{proof}
\begin{proposition}
	For any given integers $\ell$ and $n$ ($\ell,n\geq1$), consider $\x,\y\in\Sigma_{DNA}^\ell$. 
	Then, for any DNA strings $\u,\v\in\psi(\mathbb{Z}_2^n)$, the Hamming distance 
	\begin{equation*}
	H(\u,\v^{rc}) \geq  
	\left\{
	\begin{array}{ll}
	n\min\{H(\x,\y^r),H(\x,\y^{rc})\}, & \mbox{ if } n \mbox{ is even,} \\
	\min\left\{H(\x,\x^r),H(\y,\y^r),H(\x,\x^{rc}),H(\y,\y^{rc})\right\}, & \mbox{ if } n \mbox{ is odd. } 
	\end{array}\right.
	\end{equation*}
	\label{kg reverse-complement condition}
\end{proposition}
In Theorem \ref{kg reverse theorem}, a condition on DNA blocks are obtained that ensures the R constraint for the encoded DNA code.
\begin{theorem}
	For any even integer $n$ and an integer $\ell$ ($\ell,n\geq1$), if $\x,\y\in\Sigma_{DNA}^\ell$ such that $H(\x^{rc},\y)$ = $H(\x^r,\y)=\ell$ then, the DNA codes obtained from $\ell$ order Non-Homopolymer map satisfy the R and RC constraints. 
	\label{kg reverse theorem}
\end{theorem}
\begin{proof}
	If $H(\x^{rc},\y)$ = $H(\x^r,\y)=\ell$ then, from Theorem \ref{kg reverse condition}, 
	\begin{equation*}
	    \begin{split}
	        H(\u^r,\v)&\geq n\min\left\{H(\x^r,\y),H(\x^{rc},\y)\right\} \\
	        & =n\ell.
	    \end{split}
	\end{equation*}
	Similarly from Proposition \ref{kg reverse-complement condition}, 
		\begin{equation*}
	    \begin{split}
	        H(\u^{rc},\v)&\geq n\min\left\{H(\x^r,\y),H(\x^{rc},\y)\right\} \\
	        & =n\ell.
	    \end{split}
	\end{equation*}
	But the length of DNA string $\u$ and DNA string $\v$ are the same and equal to $n\ell$.
	Thus, $d_H\leq H(\u,\v)\leq n\ell$.
	Therefore, $H(\u^r,\v)\geq d_H$ and $H(\u^{rc},\v)\geq d_H$ for any DNA code obtained from $\ell$ order Non-Homopolymer map, where $H(\x^{rc},\y)$ = $H(\x^r,\y)=\ell$.
\end{proof}
\begin{example}{Example}
For $n=4$, $\ell$ = $2$, $\x$ = $AT$ and $\y$ = $CG$,
\[
\begin{array}{ccccc}
      \begin{array}{c|}
           \mathbb{Z}_2^4  \\ \hline
           (0\ 0\ 0\ 0)  \\
           (0\ 0\ 0\ 1)  \\
           (0\ 0\ 1\ 0)  \\
           (0\ 0\ 1\ 1)  \\
           (0\ 1\ 0\ 0)  \\
           (0\ 1\ 0\ 1)  \\
           (0\ 1\ 1\ 0)  \\
           (0\ 1\ 1\ 1)  \\
           (1\ 0\ 0\ 0)  \\
           (1\ 0\ 0\ 1)  \\
           (1\ 0\ 1\ 0)  \\
           (1\ 0\ 1\ 1)  \\
           (1\ 1\ 0\ 0)  \\
           (1\ 1\ 0\ 1)  \\
           (1\ 1\ 1\ 0)  \\
           (1\ 1\ 1\ 1)  
      \end{array}
       & 
      \begin{array}{l}
           \psi(\mathbb{Z}_2^4)  \\ \hline
           \x \y \x^c \y^c  \\
           \x \y \x^c \y  \\
           \x \y \x \y  \\
           \x \y \x \y^c  \\
           \x \y^c \x \y  \\
           \x \y^c \x \y^c  \\
           \x \y^c \x^c \y^c  \\
           \x \y^c \x^c \y  \\
           \x^c \y^c \x \y  \\
           \x^c \y^c \x \y^c  \\
           \x^c \y^c \x^c \y^c  \\
           \x^c \y^c \x^c \y  \\
           \x^c \y \x^c \y^c  \\
           \x^c \y \x^c \y  \\
           \x^c \y \x \y  \\
           \x^c \y \x \y^c  
      \end{array}
      &
      \begin{array}{c}
        \u     \\ \hline
           ATCGTAGC  \\
           ATCGTACG  \\
           ATCGATCG  \\
           ATCGATGC  \\
           ATGCATCG  \\
           ATGCATGC  \\
           ATGCTAGC  \\
           ATGCTACG  \\
           TAGCATCG  \\
           TAGCATGC  \\
           TAGCTAGC  \\
           TAGCTACG  \\
           TACGTAGC  \\
           TACGTACG  \\
           TACGATCG  \\
           TACGATGC  
      \end{array}
      &
      \begin{array}{c}
        \u^r     \\ \hline
           CGATGCTA  \\
           GCATGCTA  \\
           GCTAGCTA  \\
           CGTAGCTA  \\
           GCTACGTA  \\
           CGTACGTA  \\
           CGATCGTA  \\
           GCATCGTA  \\
           GCTACGAT  \\
           CGTACGAT  \\
           CGATCGAT  \\
           GCATCGAT  \\
           CGATGCAT  \\
           GCATGCAT  \\
           GCTAGCAT  \\
           CGTAGCAT  
      \end{array}
      &
      \begin{array}{c}
        \u^{rc}     \\ \hline
           GCTACGAT  \\
           CGTACGAT  \\
           CGATCGAT  \\
           GCATCGAT  \\
           CGATGCAT  \\
           GCATGCAT  \\
           GCTAGCAT  \\
           CGTAGCAT  \\
           CGATGCTA  \\
           GCATGCTA  \\
           GCTAGCTA  \\
           CGTAGCTA  \\
           GCTACGTA  \\
           CGTACGTA  \\
           CGATCGTA  \\
           GCATCGTA  
      \end{array}
    \end{array}
\]
One can easily observe that, for any $\a,\b\in\mathbb{Z}_2^4$, 
\begin{equation*}
    \begin{split}
        & H(\psi(\a)^r,\psi(\b)) = 8 \geq H(\psi(\a),\psi(\b)) \\
        & H(\psi(\a)^{rc},\psi(\b)) = 8 \geq H(\psi(\a),\psi(\b)) \\
        & d_{NHo}(\a,\b) = H(\psi(\a),\psi(\b)) 
    \end{split}
\end{equation*}
Therefore, for any binary code $\mathscr{C}\subseteq\mathbb{Z}_2^4$, the DNA code $\psi(\mathscr{C})$ satisfies R and RC constraints.
\end{example}

Now, the isometry is established between DNA codes and binary codes in the Theorem \ref{distance preserving theorem}.
\begin{theorem}
	For any integers $\ell$ and $n$ ($\ell,n\geq1$), the map 
	\[ 
	\psi:(\mathbb{Z}_2^n,d_{NHo}) \rightarrow (\psi(\mathbb{Z}_2^n), d_H)
	\] is an isometry.
	\label{distance preserving theorem}
\end{theorem}
\begin{proof} 
The result is proved using Mathematical Induction on the string length $n$. 
\begin{description}[\hspace{2.1cm}]
\item[Base case:]  For $n=1$, consider $a,b\in\mathbb{Z}_2$.
Now, one can computationally verify that 
\[
d_{NHo}(a,b)=H(\psi(a),\psi(b)).
\] 
\item[Hypothesis:] For $n=m$ and $\a,\b\in\mathbb{Z}_2^m$, assume 
\[
d_{NHo}(\a,\b)=H(\psi(\a),\psi(\b)).
\]
\item[Inductive Step:] Consider binary strings $\a$ = $(a_1\ a_2\ldots a_m)$ and $\b$ = $(b_1\ b_2\ldots b_m)$ of length $m$ with the support set $S$ and the extended support set $T$.
The binary strings are encoded into DNA strings $\psi(\a)$ = $\u$ = $ u_1  u_2\ldots  u_m$ and $\psi(\b)$ = $\v$ = $ v_1 v_2\ldots  v_m$ using $\ell$ order Non-Homopolymer map for $\x,\y\in\Sigma_{DNA}^\ell$. 
For $n=m+1$, consider the binary strings $\a^*$ = $(\a\ a_{m+1})$ = $(a_1\ a_2\ldots a_m\ a_{m+1})$ and $\b^*$ = $(\b\ b_{m+1})$ = $(b_1\ b_2\ldots b_m\ b_{m+1})$ of length $m+1$ with the support set $S^*$ and the extended support set $T^*$, where $a_{m+1},b_{m+1}\in\mathbb{Z}_2$.
For the binary strings $\a^*$ and $\b^*$, consider the DNA strings $\psi(\a^*)$ = $\u^*$ = $\u u_{m+1}$ = $ u_1  u_2\ldots  u_m  u_{m+1}$ and $\psi(\b^*)$ = $\v^*$ = $\v v_{m+1}$ = $ v_1  v_2\ldots  v_m  v_{m+1}$, where $ u_{m+1}, v_{m+1}\in\{\x,\x^c,\y,\y^c\}$. 
Now, for the binary strings $\a^*$ and $\b^*$, the support set and extended support set are 
\begin{equation*}
S^* =
\begin{cases}
S & \mbox{ if }a_m = b_m, \\
S\cup\{|S|+1\} & \mbox{ if } a_m \neq b_m,
\end{cases}
\end{equation*}
and 
\begin{equation*}
T^* =
\begin{cases}
S & \mbox{ if }a = b\mbox{ and }|S|\mbox{ is even,} \\
S\cup\{|S|+2\} & \mbox{ if }a = b\mbox{ and }|S|\mbox{ is odd,} \\
S\cup\{|S|+1,|S|+2\} & \mbox{ if }a\neq b\mbox{ and }|S|\mbox{ is even, and} \\
S\cup\{|S|+1\} & \mbox{ if }a\neq b\mbox{ and }|S|\mbox{ is odd.} 
\end{cases}
\end{equation*}
Now, from Remark \ref{distance remark 1} and Remark \ref{distance remark 2}, one can get $d_{NHo}(\a^*,\b^*)$ = $H(\psi(\a^*),\psi(\b^*))$ for various cases.
It is interesting task to identify those four cases and verify $d_{NHo}(\a^*,\b^*)$ = $H(\psi(\a^*),\psi(\b^*))$ for all the cases. 
Now, from the verification, the hypothesis holds for $n$ = $m+1$. 
\end{description}
Hence, the result follows from Mathematical Induction on the parameter $n$.
\end{proof}
\begin{example}{Example}
For each $\a,\b\in\mathbb{Z}_2^3$ and given integer $\ell$ ($\geq1$), the distance $d_{NHo}(\a,\b)$ is calculated as following. 
\[
      \begin{array}{c|cccccccc}
           d_{NHo}(\a,\b) & (0\ 0\ 0) & (0\ 0\ 1) &  (0\ 1\ 0) & (0\ 1\ 1) & (1\ 0\ 0) & (1\ 0\ 1) & (1\ 1\ 0) & (1\ 1\ 1)   \\ \hline
           (0\ 0\ 0)      & 0         & 1\ell     & 2\ell      & 1\ell     & 3\ell     & 2\ell     & 1\ell     &  2\ell      \\
           (0\ 0\ 1)      & 1\ell     & 0         & 1\ell      & 2\ell     & 2\ell     & 3\ell     & 2\ell     &  1\ell      \\
           (0\ 1\ 0)      & 2\ell     & 1\ell     & 0          & 1\ell     & 1\ell     & 2\ell     & 3\ell     & 2\ell       \\
           (0\ 1\ 1)      & 1\ell     & 2\ell     & 1\ell      & 0         & 2\ell     & 1\ell     & 2\ell     & 3\ell       \\
           (1\ 0\ 0)      & 3\ell     & 2\ell     & 1\ell      & 2\ell     & 0         & 1\ell     & 2\ell     & 1\ell       \\
           (1\ 0\ 1)      & 2\ell     & 3\ell     & 2\ell      & 1\ell     & 1\ell     & 0         & 1\ell     & 2\ell       \\
           (1\ 1\ 0)      & 1\ell     & 2\ell     & 3\ell      & 2\ell     & 2\ell     & 1\ell     & 0         & 1\ell       \\
           (1\ 1\ 1)      & 2\ell     & 1\ell     & 2\ell      & 3\ell     & 1\ell     & 2\ell     & 1\ell     & 0
      \end{array}
\]
For any $\x,\y\in\Sigma_{DNA}^\ell$, the Hamming distance $H(\psi(\a),\psi(\b))$ is calculated as following. 
\[
      \begin{array}{c|cccccccc}
           H(\psi(\a),\psi(\b)) & \x\y\x^c & \hspace{3mm}\x\y\x &  \hspace{3mm}\x\y^c\x & \hspace{3mm}\x\y^c\x^c & \hspace{3mm}\x^c\y^c\x & \hspace{3mm}\x^c\y^c\x^c & \hspace{3mm}\x^c\y\x^c & \hspace{3mm}\x^c\y\x \\ \hline
           \x\y\x^c           & 0         & 1\ell     & 2\ell      & 1\ell     & 3\ell     & 2\ell     & 1\ell     &  2\ell      \\
           \x\y\x             & 1\ell     & 0         & 1\ell      & 2\ell     & 2\ell     & 3\ell     & 2\ell     &  1\ell      \\
           \x\y^c\x           & 2\ell     & 1\ell     & 0          & 1\ell     & 1\ell     & 2\ell     & 3\ell     & 2\ell       \\
           \x\y^c\x^c         & 1\ell     & 2\ell     & 1\ell      & 0         & 2\ell     & 1\ell     & 2\ell     & 3\ell       \\
           \x^c\y^c\x         & 3\ell     & 2\ell     & 1\ell      & 2\ell     & 0         & 1\ell     & 2\ell     & 1\ell       \\
           \x^c\y^c\x^c       & 2\ell     & 3\ell     & 2\ell      & 1\ell     & 1\ell     & 0         & 1\ell     & 2\ell       \\
           \x^c\y\x^c         & 1\ell     & 2\ell     & 3\ell      & 2\ell     & 2\ell     & 1\ell     & 0         & 1\ell       \\
           \x^c\y\x           & 2\ell     & 1\ell     & 2\ell      & 3\ell     & 1\ell     & 2\ell     & 1\ell     & 0
      \end{array}
\]
Recall that, for any $\x,\y\in\Sigma_{DNA}^\ell$, $H(\x,\x^c)$ = $H(\y,\y^c)$ = $\ell$.
Hence, it is clear the $d_{NHo}(\a,\b)$ = $H(\psi(\a),\psi(\b))$ for each $\a,\b\in\mathbb{Z}_2^3$.
\end{example}

The parameters of DNA codes obtained from any given binary codes are given in Theorem \ref{kg DNA code parameters}.
\begin{theorem}
For any $(n, M, d_{NHo})$ binary code $\mathscr{C}$, an $(n\ell, M, d_H)$ DNA code $\psi(\mathscr{C})$ exists, where $d_H=d_{NHo}$.
\label{kg DNA code parameters}
\end{theorem}
\begin{proof}
The result is obtained from Theorem \ref{distance preserving theorem} and the definition of $\ell$ order Non-Homopolymer map.
\end{proof}

\subsection{Constructions of DNA Codes}
From Theorem \ref{kg DNA code parameters}, for suitable $\x,\y\in\Sigma_{DNA}^\ell$, DNA codes can be obtained from any binary codes that satisfy
\begin{itemize}
\item Tandem-free constraint with repeat-length $\lfloor n/2\rfloor$,
\item Hamming constraint,
\item R constraint, 
\item RC constraint, and 
\item $GC$-content constraint.
\end{itemize}
Thus, in this section, all the DNA codes discussed in this section satisfy all these properties togeter. 
For example, as given in \cite[Table 4]{Benerjee2021}, one can get 
\begin{itemize}
    \item $(n\ell,2,n\ell)$ DNA code from the binary code $\{(0\ \x),(1\ \x)\}$ for any given $\x\in\mathbb{Z}_2^{n-1}$, 
    \item $(4\ell,2,2\ell)$ DNA code from $[4,1,4]$ repetition code, 
    \item $(7\ell,16,2\ell)$ DNA code from $[7,4,3]$ Hamming code,  
    \item $(15\ell, 256, 3\ell)$ DNA code from $(15, 256, 5)$ Nordstrom-Robinson code, and 
    \item $(23\ell,4096,4\ell)$ DNA code from $[23,12,7]$ Golay code.
\end{itemize}
In particular, for $\ell$ = $2$, if $\x$ = $AT$ and $\y$ = $CG$ then,
\begin{itemize}
    \item from the binary code $\{(0\ 0\ 1\ 0),(1\ 0\ 1\ 0)\}$, one can get the $(6,2,6)$ DNA code with the DNA codewords
\[
\begin{array}{cll}
            \psi((0\ 0\ 1\ 0)) =  & \x\y\x\y & = ATCGTAGC,  \mbox{ and }  \\  
			\psi((1\ 0\ 1\ 0)) =  & \x^c\y^c\x^c\y^c & = TACGATGC.  
\end{array}
\]
    \item from the $[4,1,4]$ binary repetition code, one can get the $(8,2,4)$ DNA code with the DNA codewords
\[
\begin{array}{cll}
            \psi((0\ 0\ 0\ 0)) =  & \x\y\x^c\y^c & = ATCGTAGC,  \mbox{ and }  \\  
			\psi((1\ 1\ 1\ 1)) =  & \x^c\y\x\y^c & = TACGATGC.  
\end{array}
\]
    \item from the $[7,4,3]$ binary Hamming code, one can get the $(21,16,6)$ DNA code with the DNA codewords
\[
\begin{array}{cll}
            \psi((0\ 0\ 0\ 0\ 0\ 0\ 0)) =  & \x\y\x^c\y^c\x\y\x^c         & =  ATCGTAGCATCGTA,  \\  
			\psi((1\ 1\ 1\ 0\ 0\ 0\ 0)) =  & \x^c\y\x\y\x^c\y^c\x         & =  TACGATCGTAGCAT,  \\  
			\psi((1\ 0\ 0\ 1\ 1\ 0\ 0)) =  & \x^c\y^c\x\y^c\x^c\y^c\x     & =  TAGCATGCTAGCAT,  \\  
			\psi((0\ 1\ 1\ 1\ 1\ 0\ 0)) =  & \x\y^c\x^c\y\x\y\x^c         & =  ATGCTAGCATCGTA,  \\  
			\psi((0\ 1\ 0\ 1\ 0\ 1\ 0)) =  & \x\y^c\x\y^c\x\y^c\x         & =  ATGCATGCATGCAT,  \\  
			\psi((1\ 0\ 1\ 1\ 0\ 1\ 0)) =  & \x^c\y^c\x^c\y\x^c\y\x^c     & =  TAGCTACGTACGTA,  \\  
			\psi((1\ 1\ 0\ 0\ 1\ 1\ 0)) =  & \x^c\y\x^c\y\x\y^c\x         & =  TACGTACGATGCAT,  \\  
			\psi((0\ 0\ 1\ 0\ 1\ 1\ 0)) =  & \x\y\x\y\x\y^c\x             & =  ATCGATCGATGCAT,  \\  
			\psi((1\ 1\ 0\ 1\ 0\ 0\ 1)) =  & \x^c\y\x^c\y\x^c\y^c\x^c     & =  TACGTACGTAGCTA,  \\  
			\psi((0\ 0\ 1\ 1\ 0\ 0\ 1)) =  & \x\y\x\y^c\x\y\x             & =  ATCGATGCATCGAT,   \\  
			\psi((0\ 1\ 0\ 0\ 1\ 0\ 1)) =  & \x\y^c\x\y\x\y\x             & =  ATGCATGCATCGAT,  \\  
			\psi((1\ 0\ 1\ 0\ 1\ 0\ 1)) =  & \x^c\y^c\x^c\y^c\x^c\y^c\x^c & =  TAGCTAGCTAGCTA,  \\  
			\psi((1\ 0\ 0\ 0\ 0\ 1\ 1)) =  & \x^c\y^c\x\y\x^c\y\x         & =  TAGCATCGTACGAT,  \\  
			\psi((0\ 1\ 1\ 0\ 0\ 1\ 1)) =  & \x\y^c\x^c\y^c\x\y^c\x^c     & =  ATGCTAGCATGCTA,  \\  
			\psi((0\ 0\ 0\ 1\ 1\ 1\ 1)) =  & \x\y\x^c\y\x\y^c\x^c         & =  ATCGTACGATGCTA, \mbox{ and } \\
			\psi((1\ 1\ 1\ 1\ 1\ 1\ 1)) =  & \x^c\y\x\y^c\x\y\x^c         & =  TACGATGCATCGTA.  
\end{array}
\]
\end{itemize}

\section{Algebraic Bounds on DNA Codes}\label{my sec:6}
All the notations used in this section is defined as follows. 
For the given length $n$ and the minimum Hamming distance $d_H$, 
\begin{description}[\hspace{3.3cm}]
    \item[$A_2(n,d_H)$:] The maximum size of the binary code.
    \item[$A_2(n,d_H,w)$:] The maximum size of the binary constant weight code, where each codeword has the Hamming weight $w$.
    \item[$A_3(n,d_H,w)$:] The maximum size of the ternary constant weight code, where each codeword has the Hamming weight $w$.
    \item[$A_4(n,d_H)$:] The maximum size of the DNA code.
    \item[$A_2^r(n,d_H)$:] The maximum size of the binary code, where the DNA code satisfies R constraint. 
    \item[$A_2^r(n,d_H,w)$:] The maximum size of the binary constant weight code, where each codeword has the Hamming weight $w$ and the binary code satisfies the R constraint.
    \item[$A_3^r(n,d_H,w)$:] The maximum size of the ternary constant weight code, where each codeword has the Hamming weight $w$ and the ternary code satisfies the R constraint.
	\item[$A_4^r(n,d_H)$:] The maximum size of the DNA code, where the DNA code satisfies R constraint.  
	\item[$A_4^{rc}(n,d_H)$:] The maximum size of the DNA code, where the DNA code satisfies RC constraint. 
	\item[$A_4^{GC}(n,d_H,w)$:] The maximum size of the DNA code with $GC$-weight $w$, where the DNA code satisfies fixed $GC$-content constraint with weight $w$. 
	\item[$A_4^{r,GC}(n,d_H,w)$:] The maximum size of the DNA code with $GC$-weight $w$, where the DNA code satisfies R constraint and fixed $GC$-content constraint with weight $w$.  
	\item[$A_4^{rc,GC}(n,d_H,w)$:] The maximum size of the DNA code with $GC$-weight $w$, where the DNA code satisfies RC constraint and fixed $GC$-content constraint with weight $w$.   
	\item[$A_4^{r,rc}(n,d_H)$:] The maximum size of the DNA code, where the DNA code satisfies R and RC constraints. 
	\item[$A_4^{r,rc,GC}(n,d_H,w)$:] The maximum size of the DNA code with $GC$-weight $w$, where the DNA code satisfies R constraint, RC constraint and fixed $GC$-content constraint with weight $w$.  
	\item[$A_4^{GC,Homo}(n,d_H,w)$:] The maximum size of the DNA code with and $GC$-weight $w$, where the DNA code satisfies fixed $GC$-content constraint with weight $w$, and each DNA codeword is free from Homopolymers. 
\end{description}
Now, from the literature, the bounds on DNA codes with various constraints are following.
\begin{enumerate}
    \item \cite[Theorem 3.1]{doi:10.1089/10665270152530818} (Sphere-Packing bound): For given integer $n$ and $1\leq d_H\leq n$, 
    \[
    A_4(n,d_H)\leq\frac{4^n}{\sum_{i=0}^{\lfloor (d_H-1)/2\rfloor}\binom{n}{i}3^i}.
    \]
    \item \cite[Theorem 3.2]{doi:10.1089/10665270152530818} (Gilbert–Varshamov bound): For given integer $n$ and $1\leq d_H\leq n$,
    \[
    A_4(n,d_H)\geq\frac{4^n}{\sum_{i=0}^{d_H-1}\binom{n}{i}3^i}.
    \]
    \item \cite[Theorem 3.3]{doi:10.1089/10665270152530818} (Singleton bound): For given integer $n$ and $1\leq d_H\leq n$,
    \[
    A_4(n,d_H)\leq 4^{n-d_H+1}.
    \]
    \item \cite[Theorem 3.4]{doi:10.1089/10665270152530818} (Plotkin bound): For given integer $n$ and $3n/2<d_H\leq n$,
    \[
    A_4(n,d_H)\leq\frac{4d_H}{4d_H-3n}.
    \]
    \item \cite[Theorem 3.5]{doi:10.1089/10665270152530818} For given integer $n$ and $1\leq d_H\leq n$,
    \begin{itemize}
        \item  $A_4(n,d_H)\geq A_4(n+1,d_H+1)$, and
        \item  $A_4(n,d_H)\geq A_4(n+1,d_H)/4$.
    \end{itemize}
    \item \cite[Theorem 4.1]{doi:10.1089/10665270152530818} For given even integer $n$ and $1\leq d_H\leq n$, 
    \[
    A_4^{rc}(n,d_H) = A_4^r(n,d_H).
    \]
    \item \cite[Theorem 4.1]{doi:10.1089/10665270152530818} For given odd integer $n$ and $1\leq d_H\leq n$,
    \[
    A_4^r(n,d_H+1)\leq A_4^{rc}(n,d_H)\leq A_4^r(n,d_H-1).
    \]
    \item \cite[Proposition 2]{GABORIT200599} For given odd integer $n$ and $1\leq d_H\leq n$,
    \[
    A_4^{rc}(n,d_H)\leq A_4^r(n,d_H)/2.
    \]
    \item \cite[Theorem 4.3]{doi:10.1089/10665270152530818} For given integer $n$ and $1\leq d_H\leq n$, consider a set $S$ of all DNA strings of length $n$ such that, for any $\x,\y\in S$, $H(\x,\y^r)\geq d_H$ and $\x\neq\y$. 
    Then, 
        \[
        A_4^r(n,d_H)\geq\frac{4^{\lceil n/2\rceil}}{2V^+(d_H-1)}\sum_{i=\lceil d_H/2\rceil}^{\lfloor n/2\rfloor}\binom{\lfloor n/2\rfloor}{i}3^i,
        \]
        where $V^+(d_H)$ is the maximum size of the set $S$ for given $d_H$. 
        \item \cite[Theorem 4.4]{doi:10.1089/10665270152530818} (Halving bound): For given integer $n$ and $1\leq d_H\leq n$,
         \[
         A_4^r(n,d_H)\leq A_4(n,d_H)/2,
         \]
        where, for the $(n,A_4^r(n,d_H),d_H)$ DNA code $\mathscr{C}_{DNA}$, if $\x\in\mathscr{C}_{DNA}$ then $\x^r\notin\mathscr{C}_{DNA}$.
        \item \cite[Theorem 4.5]{doi:10.1089/10665270152530818} (Cai's lower bound): For given  integer $n$ and $1\leq d_H\leq n$,
        \[
        A_4^r(2n,2d_H)\geq \lfloor A_4(n,d_H)/2\rfloor.
        \]
        \item \cite[Theorem 4.7]{doi:10.1089/10665270152530818} (Product bound): For given integer $n$ and $1\leq d_H\leq n$,
        \[
        A_4^r(n,d_H)\geq A_2^r(n,d_H)\cdot A_2(n,d_H).
        \]        
        \item \cite[Theorem 4.9]{doi:10.1089/10665270152530818} For given integer $n$ and $1\leq d_H\leq n$,
        \begin{itemize}
            \item $A_4^r(n,d_H)\leq A_4^r(n,d_H-1)$, and 
            \item $A_4^r(n,d_H)/4\leq A_4^r(n-1,d_H)\leq A_4^r(n,d_H)$ for odd $n$.
        \end{itemize}
    \item \cite[Proposition 5]{GABORIT200599} For given odd integer $n$ and $1\leq d_H,w\leq n$,
    \[
    A_4^{rc,GC}(n,d_H,w)\leq A_4^{r,GC}(n,d_H,w)/2.
    \]
    \item \cite[Proposition 9]{GABORIT200599} For given integer $n$ and $1\leq d_H,w\leq n$,
    \[
    A_4^{rc,GC}(n,d_H,w)\geq A_2^r(n,d_H,w)\cdot A_2(n,d_H).
    \]
    \item \cite[page no. 110]{GABORIT200599} For given integer $n$ and $1\leq d_H,w\leq n$,
    \begin{itemize}
        \item $A_4^{rc,GC}(n,d_H,w)\leq A_4^{rc,GC}(n,d_H-1,w)$, and 
        \item $A_4^{rc,GC}(n,d_H,w)\leq A_4^{rc,GC}(n+1,d_H,w)$.
    \end{itemize}
    \item \cite[Proposition 1]{King2003} For given integer $n$ and $1\leq d_H,w\leq n$,
    \begin{itemize}
        \item $A_4^{GC}(n,d_H,w)=A_4^{GC}(n,d_H,n-w)$, and
        \item $A_4^{GC}(n,d_H,0)=A_2(n,d_H)$.
    \end{itemize}
    \item \cite[Theorem 2]{King2003} (Johnson-type bound): For given integer $n$ and $1\leq d_H,w\leq n$,
    \begin{itemize}
        \item $A_4^{GC}(n,d_H,w)\leq\left\lfloor\frac{2n}{w}A_4^{GC}(n-1,d_H,w-1)\right\rfloor$, and 
        \item $A_4^{GC}(n,d_H,w)\leq\left\lfloor\frac{2n}{n-w}A_4^{GC}(n-1,d_H,w)\right\rfloor$.
    \end{itemize}
    \item \cite[Theorem 5]{King2003} For given integer $n$ and $1\leq d_H,w\leq n$, if $2nd_H>n^2+2nw-2w^2$ then 
    \[
    A_4^{GC}(n,d_H,w)\leq\frac{2nd_H}{2nd_H-(n^2+2nw-2w^2)}.
    \] 
    \item \cite[Theorem 8]{King2003} (Gilbert-type bound): For given integer $n$ and $1\leq d_H,w\leq n$, 
    \[
    A_4^{GC}(n,d_H,w)\geq\frac{\binom{n}{w}2^n}{\sum_{r=0}^{d_H-1}\sum_{i=0}^{\min\{\lfloor r/2\rfloor,w,n-w\}}\binom{w}{i}\binom{n-w}{i}\binom{n-2i}{r-2i}2^{2i}}.
    \]
    \item \cite[Theorem 11]{King2003} (Gilbert-type bound): For given integer $n$ and $1\leq d_H,w\leq n$, 
    \[
    A_4^{rc,GC}(n,d_H,w)\geq\frac{\sum_{r=d_H}^nV(n,r,w)}{2\sum_{r=0}^{d_H-1}\sum_{i=0}^{\min\{\lfloor r/2\rfloor,w,n-w\}}\binom{w}{i}\binom{n-w}{i}\binom{n-2i}{r-2i}2^{2i}},
    \]
    where $V(n,r,w)$ is the size of the set 
    \[
    \{\x:H(\x,\x^{rc})=r\mbox{ and }w_{GC}(\x)=w\mbox{ for }\x\in\Sigma_{DNA}^n\}.
    \]
    \item \cite[Proposition 12]{King2003} For given integer $n$ and $1\leq d_H,w\leq n$,  
    \begin{itemize}
        \item $A_4^{rc,GC}(n,d_H,w)$ = $A_4^{r,GC}(n,d_H,w)$ for even $n$, and 
        \item $A_4^{r,GC}(n,d_H+1,w) \leq A_4^{rc,GC}(n,d_H,w)\leq A_4^{r,GC}(n,d_H-1,w)$ for odd $n$.
    \end{itemize}
    \item \cite[Theorem 13]{King2003} For given integer $n$ and $1\leq d_H,w\leq n$,  
    \begin{itemize}
        \item $A_4^{GC}(n,d_H,w)\geq A_2(n,d_H,w)\cdot A_2(n,d_H)$,  
        \item $A_4^{r,GC}(n,d_H,w)\geq A_2^r(n,d_H,w)\cdot A_2(n,d_H)$,
        \item $A_4^{r,GC}(n,d_H,w)\geq A_2(n,d_H,w)\cdot A_2(n,d_H)^r$,
        \item $A_4^{GC}(n,d_H,w)\geq A_3(n,d_H,w)\cdot A_2(n-w,d_H)$,
        \item $A_4^{r,GC}(n,d_H,w)\geq A_3^r(n,d_H,w)\cdot A_2(n-w,d_H)$, and
        \item $A_4^{r,GC}(n,d_H,w)\geq A_3(n,d_H,w)\cdot A_2^r(n-w,d_H)$.
    \end{itemize}
    \item \cite[Theorem 2]{8424143} For given integer $n$ and $1\leq d_H,w\leq n$, 
    \[
    A_4^{GC,Homo}(n,d_H,w)\geq\frac{B(n,w)}{\sum_{r=0}^{d_H-1}\sum_{i=0}^{\min\{\lfloor r/2\rfloor,w,n-w\}}\binom{w}{i}\binom{n-w}{i}\binom{n-2i}{r-2i}2^{2i}},
    \]
    where
    \[
    B(n,w) = \sum_{j=0}^{v-1}2^{2v+1-2j}\binom{v-1}{j}\binom{n-v}{v-j}+\sum_{j=0}^{v-2}2^{2v-1-2j}\binom{v-1}{j}\binom{n-v-1}{v-j-2}, 
    \]
    and $v$ = $\min\{w,n-w\}$.
\end{enumerate}

\section{Some Open Problems}
\label{my sec:7}
The designing of DNA codes with the desired properties is somewhat still an open challenge despite of so much literature. In this chapter, we presented an algebraic 
approach for the construction of DNA codes. We summarise the following research directions that one can explore further. 
%
\begin{description}[\hspace{2cm}]
\item[Problem \ref{my sec:7}.1] Exploring algebraic structures such as other finite rings and finite fields that can yield DNA codes with high minimum Hamming distance.
\item[Problem \ref{my sec:7}.2] Developing techniques for handling new constraints (such as secondary structure formation) via algebraic means arising from DNA storage applications.
\item[Problem \ref{my sec:7}.3] Using computational tools such as Magma together with codes over finite algebraic structures and computational techniques in constructing large set of DNA codes. 
\item[Problem \ref{my sec:7}.4] Updating the Tables of DNA codes by filling the gaps.
\item[Problem \ref{my sec:7}.5] Finding tight bounds on DNA codes with various constraints and properties.
\item[Problem \ref{my sec:7}.6] Finding optimal codes (bounds achieving) DNA codes with various constraints and properties.
\end{description}

\bibliographystyle{plain}
\bibliography{DNA}

\end{document}